\documentclass{article}

\usepackage{hyperref}
\usepackage{amsmath}
\usepackage{amssymb}
\usepackage{amsthm}
\usepackage{vmargin}

\setpapersize{USletter}
\setmarginsrb{1.3in}{1in}{1.3in}{1.5in}{0cm}{0cm}{0cm}{.5in}

\makeatletter

\gdef\@runningfoot{\relax}
\gdef\@copyrightpermission{\relax}
\def\runningfoot#1{\relax}

\newenvironment{enumlist}{%
\leftmargini32pt
   \begin{itemize}
}{%
   \end{itemize}}
\makeatother

\pdfmapline{=xbbold <./xbbold.pfb}

\DeclareFontFamily{U}{xbbold}{}
\DeclareFontShape{U}{xbbold}{m}{n}{<-6>xbbold<6-8>xbbold<8->xbbold}{}
\DeclareMathAlphabet{\blmathbb}{U}{xbbold}{m}{n}

\usepackage{graphicx}
\usepackage{tabu}

\begin{document}

\title{Invariance to Ordinal Transformations in Rank-Aware Databases}

\date{\normalsize%
  Dept. Computer Science, Palacky University Olomouc}

\author{Vilem Vychodil\footnote{%
    e-mail: \texttt{vychodil@binghamton.edu},
    phone: +420 585 634 705,
    fax: +420 585 411 643}}

\maketitle

\begin{abstract}
  We study influence of ordinal transformations on results of queries
  in rank-aware databases which derive their operations with ranked relations
  from totally ordered structures of scores with infima acting as aggregation
  functions. We introduce notions of ordinal containment and
  equivalence of ranked relations and prove that infima-based algebraic
  operations with ranked relations are invariant to ordinal transformations:
  Queries applied to original and transformed data yield results
  which are equivalent in terms of the order given by scores, meaning that
  top-$k$ results of queries remain the same. We show this
  important property is preserved in alternative query systems based of
  relational calculi developed in context of G\"odel logic. We comment on
  relationship to monotone query evaluation and show that the results
  can be attained in alternative rank-aware approaches.
\end{abstract}

\def\citeN#1{\cite{#1}}

\newtheorem{consequence}{Consequence}
\def\osub{\ensuremath{\sqsubseteq}}
\def\nosub{\ensuremath{\not\sqsubseteq}}
\def\oeqv{\ensuremath{\equiv}}
\def\inf{\ensuremath{\mathop{\mathrm{inf}}}}
\def\sup{\ensuremath{\mathop{\mathrm{sup}}}}
\def\lc{\ensuremath{\mathop{\mathrm{lc}}}}
\def\uc{\ensuremath{\mathop{\mathrm{uc}}}}
\let\bigwedge=\inf
\let\bigvee=\sup
\def\logneg{\ensuremath{\neg}}
\def\logand{\ensuremath{\wedge}}
\def\logor{\ensuremath{\vee}}
\def\logimp{\ensuremath{\Rightarrow}}
\def\logequ{\ensuremath{\Leftrightarrow}}
\def\False{\overline{0}}

\newcommand{\restrict}[2]{\ensuremath{\sigma_{\!#1}(#2)}}
\newcommand{\join}[2]{\ensuremath{#1 \bowtie #2}}
\newcommand{\project}[2]{\ensuremath{\pi_{\!#1}(#2)}}
\newcommand{\union}[2]{\ensuremath{#1 \cup #2}}
\newcommand{\minus}[2]{\ensuremath{#1 - #2}}
\newcommand{\frdiv}{\div}
\newcommand{\rdiv}[3]{\ensuremath{#2 \frdiv^{#1} #3}}
\newcommand{\gsubsop}{\ensuremath{\mathrm{S}}}
\newcommand{\gsubs}[2]{\ensuremath{\gsubsop(#1,#2)}}
\newcommand{\gequsop}{\ensuremath{\mathrm{E}}}
\newcommand{\gequs}[2]{\ensuremath{\gequsop(#1,#2)}}
\newcommand{\diff}{\ensuremath{\ominus}}
\def\itm#1{{\rm(\textit{\romannumeral#1})}}
\newcommand{\rcond}[1][\theta]{\ensuremath{\theta}}
\newcommand{\rsym}[1][r]{\ensuremath{\blmathbb{#1}}}
\newcommand{\resd}[3]{\ensuremath{#2 \fresd^{#1} #3}}
\newcommand{\fresd}{\ensuremath{\rightarrow}}

\newenvironment{relation}[1]
{\bgroup%
  \tabcolsep=.7ex%
  \arrayrulewidth=.8pt%
  \renewcommand{\arraystretch}{0.9}%
  \begin{tabu}{@{}|l||*{#1}{l|}@{}}\firsthline}
{\lasthline\end{tabu}\egroup}

\def\atr#1{\ensuremath{\text{\uppercase{\texttt{\textup{#1}}}}}}
\def\prop#1#2{\ensuremath{\text{\texttt{\textit{#1}}}(#2)}}
\def\pn#1{\ensuremath{\text{\texttt{\textit{#1}}}}}
\def\val#1{\ensuremath{\text{\texttt{\textup{#1}}}}}

\def\rank{\#\rule{0pt}{8pt}}
\def\usd#1{#1}

\catcode64=11
\newcounter{global}
\theoremstyle{definition}
\newtheorem{definition}[global]{Definition}
\theoremstyle{plain}
\newtheorem{theorem}[global]{Theorem}
\newtheorem{proposition}[global]{Proposition}
\newtheorem{lemma}[global]{Lemma}
\newtheorem{corollary}[global]{Corollary}
\newtheoremstyle{note}{}{}{}{}{\itshape}{.}{.5em}{}
\theoremstyle{note}
\newtheorem{remark}{Remark}%
\newtheorem{example}{Example}%
\renewcommand\section{%
  \@startsection {section}{1}{\z@}%
  {-3.5ex \@plus -1ex \@minus -.2ex}%
  {2.3ex \@plus.2ex}%
  {\normalfont\large\bfseries}}
\def\itm#1{{\rm(\textit{\romannumeral#1})}}
\newcommand{\univ}[1][\Gamma]{\ensuremath{\mathrm{Th}(#1,Y)}}
\def\lexlt{\ensuremath{\mathrel{\lhd}}}
\newcommand{\yields}[3][\infty]{\ensuremath{#2 \mathrel{\sqsupset^*_{#1}} #3}}
\catcode64=12

\section{Introduction}
In this paper, we describe invariance to ordinal transformations in query systems
which incorporate ranking of query results and allow to compare the
importance or relevance of query results based on their scores.
We present general observations which may be applied in various
rank-aware approaches, see~\cite{Il:ASOTQPTiRDS}
for an extensive and systematic survey of existing approaches. In particular,
we present detailed analysis in a particular relational query system where
queries are expressed by arbitrary complex algebraic expressions and answered
by relations with tuples annotated by scores (so-called ranked relations or
ranked data tables). While we analyze the issues of ordinal transformations in
one particular rank-aware model, the presented technique is indeed general and
can be applied to other models which we demonstrate by showing analogous
results for~RankSQL proposed by~\cite{LiChCh:RQAaOfRTQ}.

The practical contribution of the results presented in our paper is in exposing
transformations of input ranking criteria which do not alter results of queries.
Typically, rank-aware queries may be understood as classic queries which
in addition incorporate ranking criteria like ``low price'', ``high availability'',
``close distance'', etc. Such criteria may be defined in many ways
and/or may depend on various parameters. For instance a ``distance of locations''
of houses in a city may be based on a geographical distance, a road traveling
distance, or it may be based on socio-economic parameters such as
criminality rate and rating of schools, etc. Therefore, there is a natural
question whether results of queries change if the ranking criteria are altered.

To further explain the issues studied in this paper, we can consider the
motivation presented in the classic paper by~Fagin~\cite{Fa98:CFIfMS}:
We assume a query system which admits queries that can be answered by
relations with tuples annotated by scores.
We assume that the scores have \emph{comparative meaning}
(higher scores mean better matches) and users are interested in listing
query results sorted by scores with highest scores coming first.
Recall that in~\cite{Fa98:CFIfMS}, scores of results of queries which
are expressed as conjunctions of subqueries are computed using monotone
and strict aggregation functions, typically triangular norms~\cite{KMP:TN}
on the real unit interval. For instance, consider an expression
\begin{align}
  \prop{notExceeding}{\atr{price}, \val{\$800,000}}
  \,\mathop{\texttt{AND}}\,
  \prop{near}{\atr{location}, \val{"Old Palo Alto"}}
  \label{eqn:motiv}
\end{align}
which may be regarded as query for
houses sold at \val{\$800,000} (or a similar price which does
not exceed the value too much) in
\val{Old Palo Alto} (or near that general area).
The subqueries $\prop{notExceeding}{\atr{price}, \val{\$800,000}}$ and
$\prop{near}{\atr{location}, \val{"Old Palo Alto"}}$ may be understood
as restrictions using ranking criteria $\pn{notExceeding}$
and $\pn{near}$ which we further call \emph{general restriction conditions.}
The result of evaluating~\eqref{eqn:motiv} may be seen as
a ranked relation which results by first evaluating the
subqueries $\prop{notExceeding}{\atr{price}, \val{\$800,000}}$
and $\prop{near}{\atr{location}, \val{"Old Palo Alto"}}$
which produce ranked relations as the results of subqueries and then
aggregating the scores by a conjunctive aggregation function $\otimes$.
If the scores come from the real unit interval, it is natural to assume
that
\begin{align}
  a \otimes 1 &= a, \\
  a \otimes b &= b \otimes a, \\
  a \otimes (b \otimes c) &= (a \otimes b) \otimes c
\end{align}
are satisfied for all $a,b,c \in [0,1]$ and $\otimes$ is monotone (isotone)
with respect to the usual ordering of reals:
\begin{align}
  \text{if } a \leq b \text{, then } a \otimes c \leq b \otimes c
\end{align}
for all $a,b,c \in [0,1]$. Functions satisfying such conditions are
called triangular norms~\cite{KMP:TN} and may be understood
as generalizations of (truth functions of)
the classic conjunction~\cite{GaJiKoOn:RL}.

The main concern of~\cite{Fa98:CFIfMS} are algorithms
for efficient computation of top-$k$ results of queries like~\eqref{eqn:motiv}.
Interestingly, the paper shows a simplification of the main
algorithm for returning the top-$k$ answers of monotone queries in case
the utilized aggregation function $\otimes$ coincides with $\min(x,y)$.
In this paper, we show that the choice of minimum as the basic conjunctive
aggregation function has another important (and desirable) consequence:

\begin{consequence}\label{c:1}
  Top-$k$ results of queries do not change if ordinal transformations are
  applied to the input data and all restriction conditions which appear in queries.
\end{consequence}

By an ordinal transformation we mean a transformation which modifies
scores but preserves the order of tuples given by the scores
(a precise definition follows in the paper).
We call Consequence~\ref{c:1} the \emph{invariance to ordinal transformations.}
It can be shown that the consequence does not hold
in case of general aggregation functions. Therefore, Consequence~\ref{c:1} describes
an important property of a particular query system supporting top-$k$ queries.
Systems supporting top-$k$ queries are recently gaining
interest~\cite{Li:Ptkdqud,Ng:Tatkqubddm,Wa:TkqRDF} and investigations in this
direction may be exploited as optimization techniques. For instance, in case of
reiterated queries executed with different parameters, observations like 
Consequence~\ref{c:1} may help identify that certain changes in input parameters
will have no influence on the results of top-$k$ queries. Such observations are
beneficial especially in case of processing large
data collections~\cite{PhChZh:BigData}.

We investigate the invariance to ordinal transformations independently on
the chosen structure of scores. Instead of assuming a subset of reals
with its natural ordering as the set of scores, we assume that the set of
scores forms a totally ordered set where infima (greatest lower bounds)
exist for arbitrary subsets of scores. Note that in this setting, an infimum
of a finite non-empty set of scores is its minimum element with respect to
the total order of scores. Considering such general structures of scores,
we introduce algebraic operations on ranked relations. The operations include
the join, restriction, projection, union, difference, residuum, and division
and they can be seen as particularizations of operations used
in~\cite{BeVy:Qssbd} considering the operation of infimum as the
aggregation function. Such a set of operations is adequate for
formulating complex queries, including non-monotone ones.
Using the proposed operations, queries like~\eqref{eqn:motiv} may be
regarded as particular joins (or intersections) of general restrictions.
The results elaborated in Section~\ref{sec:ordeq} and
Section~\ref{sec:inv} show properties of ordinal transformations and
the invariance theorems.

Let us stress one practical aspect about the consequences of the
invariance theorems:

\begin{consequence}\label{c:2}
  When using infima-based algebraic operations with ranked relations,
  the scores in ranked tables have no quantitative meaning.
\end{consequence}

In other words, the meaning of scores is \emph{purely comparative.} For instance,
if a ranked relation consists of exactly two tuples with ranks $a$ and $b$ such
that $a < b$, then any kind of distance or closeness of $a$ and $b$ is irrelevant;
$a = 0.3$ and $b = 0.9$ represent the same relationship as
$a = 0.89$ and $b = 0.9$ because in both cases $a < b$.
In fact, considering the generality of our structures of scores,
we may replace the numerical scores by symbolic ones as long as the order
of the corresponding scores preserves (and reflects) the order of the
numerical scores. For instance, instead of numerical values $0$, $0.8$,
and $1$, one can use symbolic names ``not at all'', ``more or less'',
and ``fully'' provided that the order $<$ is defined as
$\text{``not at all''} < \text{``more or less''} < \text{``fully''}$.
As an immediate consequence for users of a database system implementing
infima-based operations is that the scores can be completely hidden from users
since their values (and their mutual similarity) do not represent
any quantitative information.

Our paper presents an order-theoretic treatment of invariance issues
in ranked-aware databases which are traditionally studied form the
point of view of query execution. Let us note that according to a
taxonomy introduced in~\cite{Il:ASOTQPTiRDS}, in this paper we work with a model
which uses \emph{exact methods over certain data.} Indeed, we would like
to stress that unlike the approaches to probabilistic databases~\cite{DaReSu:Pdditd}
which also contain explicit scores, we are concerned with \emph{certain data.}
Possible extensions of our observations to models for uncertain data should prove
interesting but it is not the objective of our paper.

The present paper is related to our previous
analysis of preservation of similarity of query results for similar
input data~\cite{BeUrVy:Sadrqlor} where the similarity is formalized as
closeness of scores without considering the issues of order preservation.
Indeed, \cite{BeUrVy:Sadrqlor} introduces formulas for expressing lower
bounds of similarity of query results performed with pairwise similar
input data. The notion of similarity (of input data and results of queries)
in~\cite{BeUrVy:Sadrqlor} is based on residuated
implications~\cite{Gog:Lic} and captures the fact that ranked relations
consist of tuples with similar scores in terms of their closeness.
As such, the notion does not express the fact that the order of tuples
is preserved. In Section~\ref{sec:inv}, we present notes on how this
approach can be combined with the present one. Also note that issues
related to ordinal transformations of object-attribute data were studied
from the point of view of formal concept analysis~\cite{GaWi:FCA}
in~\cite{Be:Oed} where the author shows that ordinally equivalent
input data induce almost isomorphic concept lattices.

Our paper is organized as follows. Section~\ref{sec:prelim} presents
preliminaries from partially ordered sets and lattices which are used
in our paper as the basic structures of scores. Section~\ref{sec:relmod}
describes the rank-aware model for which we make the analysis of the
invariance to ordinal transformations. Section~\ref{sec:ordeq} is devoted
to the properties of order equivalence of ranked relations which plays
an important role in the analysis. Section~\ref{sec:inv} contains the
invariance and additional discussion. Section~\ref{sec:heyting} shows
how our results relate to a relational calculus developed in
the context of G\"odel logic. Section~\ref{sec:comp} discusses issues
of efficient query evaluation which arise in the model and comments
on the relationship to other approaches.

\section{Preliminaries}\label{sec:prelim}
In this section, we recall preliminary notions of partially ordered sets and
lattices. The notions are used in further sections to formalize structures of
scores which are used in the considered model of data. More details on the
notions presented in this section can be found in~\cite{Bir:LT}.

A partial order on a non-empty set $L$ is a binary relation $\leq$ on $L$
which is reflexive ($a \leq a$), antisymmetric ($a \leq b$ and $b \leq a$
yield $a = b$), and transitive ($a \leq b$ and $b \leq c$ yield $a \leq c$).
A pair $\mathbf{L} = \langle L,\leq\rangle$ where $\leq$ is a partial order
on $L$ is called a partially ordered set (shortly, a poset).
A partial order $\leq$ on $L$ is called a total (or a linear) order whenever
for any $a,b \in L$, we have $a \leq b$ or $b \leq a$
in which case $\mathbf{L} = \langle L,\leq\rangle$ is called
a totally ordered set or a chain.

An element $a \in L$ is called the least element of $K \subseteq L$
in $\mathbf{L} = \langle L,\leq\rangle$ whenever $a \leq b$ for all $b \in K$.
Dually, we consider the notion of a greatest element. If both the least and
the greatest elements exist for the whole $L$, we say that $\mathbf{L}$ is
bounded and denote the fact by $\mathbf{L} = \langle L,\leq,0,1\rangle$
where $0$ and $1$ stand for the least and the greatest element
in $\mathbf{L}$, respectively.

For subsets of a partially ordered set $\mathbf{L} = \langle L,\leq\rangle$,
we consider their greatest lower bounds and least upper bounds as follows:
For $K \subseteq L$, we put
\begin{align}
  \lc K &= \{a \in L;\, a \leq b \text{ for all } b \in K\}, \\
  \uc K &= \{b \in L;\, a \leq b \text{ for all } a \in K\},
\end{align}
and call $\lc K$ and $\uc K$ the lower and upper cones of $K$
in $\mathbf{L}$, respectively. If $\lc K$ has the greatest element,
it is called the infimum (the greatest lower bound) of $K$ in $\mathbf{L}$
and denoted by $\inf K$.
Dually, if $\uc K$ has the least element, it is called the supremum
(the least upper bound) of $K$
in $\mathbf{L}$ and denoted by $\sup K$. If for any $a,b \in L$,
the elements $\inf \{a,b\}$ and $\sup \{a,b\}$ exist,
then $\leq$ is called a lattice order and $\mathbf{L}$ is
called a lattice ordered set (shortly, a lattice).
Each totally ordered set is a lattice because if $a \leq b$,
then obviously $\inf \{a,b\} = a$ and $\sup \{a,b\} = b$.
In addition, if $\mathbf{L}$ is totally ordered then for any
non-empty and finite $K$, it follows that $\inf K$ and $\sup K$
coincide with the least and greatest elements in $K$, respectively.
If for any $K \subseteq L$, the elements $\inf K$ and $\sup K$ exist,
then $\leq$ is called a complete lattice order and $\mathbf{L}$ is
called a complete lattice ordered set (shortly, a complete lattice).
Each complete lattice is a bounded lattice because
$\inf L = \sup \emptyset = 0$ (the least element of $\mathbf{L}$)
and $\inf \emptyset = \sup L = 1$ (thre greatest element of $\mathbf{L}$).
In general, a (bounded) lattice may not be complete
(consider a subset of reals $L = [0,1] \setminus \{0.5\}$ equipped with the
usual ordering $\leq$ of reals).

There is an alternative view of partially ordered sets and lattices via
algebraic structures: Let $\mathbf{L} = \langle L,\sqcap,\sqcup\rangle$
be an algebra with two binary operations $\sqcap$ (called a meet)
and $\sqcup$ (called a join) such that both $\sqcap$ and $\sqcup$ are
commutative, associative, idempotent (i.e., $a \sqcap a = a \sqcup a = a$
for any $a \in L$), and satisfy the laws of absorption:
$a \sqcap (a \sqcup b) = a$ and $a \sqcup (a \sqcap b) = a$ for all $a,b \in L$,
then the algebra $\mathbf{L}$ is called a lattice. It can be easily shown
that for a partially ordered set $\mathbf{L} = \langle L,\leq\rangle$ which
is a lattice in the order-theoretic sense, we can consider
an algebra on $L$ with $a \sqcap b = \inf \{a,b\}$ and $a \sqcup b = \sup \{a,b\}$
which is a lattice in the latter sense.
Conversely, for a lattice $\mathbf{L} = \langle L,\sqcap,\sqcup\rangle$, we may
introduce $a \leq b$ if{}f $a = a \sqcap b$ (or, equivalently, $a \sqcup b = b$).
As a consequence, we may understand lattices as both special partially ordered
structures and special algebras. In the paper, whenever we consider a (complete)
lattice $\mathbf{L}$, we automatically consider the lattice order $\leq$ and
treat $\inf$ and $\sup$ as operations on $L$.

In the follows sections, we assume that scores we use to annotate tuples in
relations come from a bounded totally ordered set
$\mathbf{L} = \langle L,\leq,0,1\rangle$ and, optionally, we assume that
$\mathbf{L}$ is in addition a complete lattice. The operations $\inf$ and $\sup$
are used to obtain greatest lower and least upper bounds of finite (or arbitrary
if $\mathbf{L}$ is complete) subsets of scores. For instance, if one considers
conjunctive (or disjunctive) queries consisting of several subqueries,
$\inf$ (or $\sup$) is used to aggregate scores from the subqueries to
obtain a score for the composed query. Details are discussed in
Section~\ref{sec:relmod}.

\begin{remark}\label{rem:Bool2}%
  As a borderline case of complete totally ordered lattices we may take the
  two-element Boolean algebra which is uniquely given up to isomorphism. That is,
  for $L = \{0,1\}$ where $0$ denotes the truth value ``false'' and $1$ denotes
  the truth value ``true'' and putting $0 \leq 1$, we obain a complete totally
  ordered lattice where $\inf$ and $\sup$ coincide with the truth functions
  of the classic logical connectives ``conjunction'' and ``disjunction''.
  In addition, a truth function of negation may be introduced as a complement,
  i.e., $0' = 1$ and $1' = 0$ and the resulting structure
  $\mathbf{L} = \langle L,\inf,\sup,',0,1\rangle$ is a two-element Boolean
  algebra. From the point of view of the ranked approach used in this paper,
  $0$ and $1$ may be seen as two borderline
  scores---$0$ represents a mismatch while $1$ represents a match. Because of
  the well-known property of functional completeness of Boolean algebras,
  every $n$-ary truth function may be expressed by means of terms consisting
  of variables and $\inf$, $\sup$, and $'$. For instance, a truth function for
  implication (logical conditional) can be introduced as
  $a \rightarrow b = \sup\{a', b\}$.
\end{remark}

Considering Remark~\ref{rem:Bool2} and the fact that general bounded totally
ordered sets serve as structures of ranks, $\inf$ and $\sup$ may be seen as
generalizations of truth functions of logical connectives ``conjunction'' and
``disjunction''. A natural question is whether we can obtain analogies of
truth functions of other important locical connectives like the negation and
implication. This question is important because in the classic relational model
such connectives as crucial for expressing many relational operations like the
difference, semidifference, and division which cannot be expressed just using
$\inf$ and $\sup$. We therefore consider additional binary connectives which
are adjoint to $\inf$ and $\sup$ and serve as generalizations of truth functions
of logical connectives ``implication'' and ``non-implication''
(so-called abjunction): For a bounded lattice
$\mathbf{L} = \langle L,\leq,0,1\rangle$, consider a binary operaton $\rightarrow$
such that 
\begin{align}
  \inf \{a,b\} \leq c \text{ if{}f } a \leq b \rightarrow c
  \label{eqn:adj}
\end{align}
holds true for all $a,b,c \in L$.
Note that $\rightarrow$ may not exist but if
it exists for given $\mathbf{L}$, then it is given uniquely. In the terminology
or ordered sets, $\rightarrow$ is called a relative pseudo-complement or
a residuum.
Alternatively, $\rightarrow$ may be introduced as a binary operation
on $\mathbf{L}$ which satisfies the following conditions
\begin{align}
  a \rightarrow a &= 1, \\
  \inf \{a, a \rightarrow b\} &= \inf \{a,b\}, \\
  \inf \{a \rightarrow b, b\} &= b, \\
  a \rightarrow \inf\{b,c\} &= \inf \{a \rightarrow b, a \rightarrow c\},
\end{align}
for every $a,b,c \in L$. The resulting structure 
$\mathbf{L} = \langle L,\inf,\sup,\rightarrow,0,1\rangle$ is called
a~Heyting algebra. Note that Heyting algebras are used as semantic structures
of the intuitionistic logic~\cite{He}. That is, in the intuitionistic logic, they play
an analogous role as the Boolean algebras in the classic logic.
If $\mathbf{L}$ satisfies the following additional condition
\begin{align}
  \sup \{a \rightarrow b, b \rightarrow a\} &= 1
\end{align}
for all $a,b \in L$, then it is called a G\"odel algebra. Analogously
as the Heyting algebras are the semantic structures of the intuitionistic logic,
G\"odel algebras are semantic structures of G\"odel logic which is a stronger
logic than the intuitionistic logic but it is not as strong as the Boolean logic.
According to~\cite{Haj:MFL}, G\"odel logic may be seen as a schematic
extension of the Basic logic.

Condition~\eqref{eqn:adj} is called the adjointness property of $\inf$
and $\rightarrow$. It ensures that $\rightarrow$ is a faithful truth function
of a general implication. In particular, if $\mathbf{L}$ is totally ordered,
we get that
\begin{align}
  a \rightarrow b &=
  \begin{cases}
    1, & \text{if } a \leq b, \\
    b, & \text{otherwise,}
  \end{cases}
  \label{eqn:resd_chain}
\end{align}
for all $a,b \in L$. Therefore, if we restrict ourselves just to
$\{0,1\} \subseteq L$, we get $1 \rightarrow 0 = 0$ and
$0 \rightarrow 0 = 0 \rightarrow 1 = 1 \rightarrow 1 = 1$, i.e.,
on $\{0,1\}$, $\rightarrow$ acts as a truth function of the classic
implication. In general, we have $a \rightarrow b = 1$ if{}f $a \leq b$.
Now, a generalization of the classic negation and equivalence
(logical biconditional) can be introduced by
\begin{align}
  \neg a &= a \rightarrow 0, \\
  a \leftrightarrow b &= \inf\{a \rightarrow b, b \rightarrow a\},
\end{align}
for all $a,b \in L$. Taking~\eqref{eqn:resd_chain} into account, we have 
\begin{align}
  \neg a &=
  \begin{cases}
    1, & \text{for } a = 0, \\
    0, & \text{otherwise,}
  \end{cases}
  \label{eqn:neg_chain}
  \\
  a \leftrightarrow b &=
  \begin{cases}
    1, & \text{for } a = b, \\
    \inf\{a,b\}, & \text{otherwise,}
  \end{cases}
  \label{eqn:neg_biresd}
\end{align}
where the $\inf\{a,b\}$ in ~\eqref{eqn:neg_biresd} is in fact a minimum of
$a$ and $b$ since all elements in a totally ordered $\mathbf{L}$ are comparable.

Analogously as $\rightarrow$ is adjoint to $\inf$ in sense of~\eqref{eqn:adj}, we
may apply the duality principle and introduce a generalization of logical
non-implication which is adjoint to $\sup$. Namely, following the ideas
of~\cite{Ra:SBataildo}, see also~\cite{OrRa:Rmcs}, we may consider
a binary operation $\diff$ such that
\begin{align}
  a \diff b \leq c \text{ if{}f } a \leq \sup\{b,c\}
  \label{eqn:dadj}
\end{align}
for all $a,b,c \in L$. Alternatively, we may postulate the following equalities:
\begin{align}
  a \diff a &= 0, \\
  \sup \{a \diff b, b\} &= \sup \{a,b\}, \\
  \sup \{a, a \diff b\} &= a, \\
  \sup\{a,b\} \diff c  &= \sup \{a \diff c, b \diff c\},
\end{align}
for every $a,b,c \in L$. Analogously as in the case of $\rightarrow$,
if $\mathbf{L}$ is totally ordered, it follows that
\begin{align}
  a \diff b &=
  \begin{cases}
    0, & \text{if } a \leq b, \\
    a, & \text{otherwise,}
  \end{cases}
  \label{eqn:diff_chain}
\end{align}
for all $a,b \in L$. Therefore, $\diff$ is indeed a generalization of
a non-implication (a logical difference bounded by $0$ and $1$) because
on $\{0,1\} \subseteq L$, we have
$0 \diff 0 = 0 \diff 1 = 1 \diff 1 = 0$ and
$1 \diff 0 = 1$. In the following sections, we utilize $\rightarrow$
in the definitions of rank-aware relational containment and division
and $\diff$ is utilized in a rank-aware relational difference.

If $\mathbf{L}$ is a totally ordered G\"odel algebra which is defined
on the real unit interval (with $\leq$ being the usual ordering of reals),
then we call $\mathbf{L}$ the standard G\"odel algebra~\cite{Haj:MFL}
and denote it by $[0,1]_\mathbf{G}$.

\section{Rank-Aware Relational Model of Data}\label{sec:relmod}
In this section, we introduce a relational rank-aware model of data. Namely,
we describe structures formalizing data tables which appear in the model and
relational operations which constitute the core of relational queries that
take scores into account. The model may be viewed as a particularization of
a model based on complete residuated lattices which has been outlined
in~\cite{BeVy:Qssbd} which results by a choice of special structures of
scores based on G\"odel algebras described in Section~\ref{sec:prelim}.

First, we recall the basic notions which appear in the (classic) relational
model of data. In the paper we consider \emph{relation schemes} as finite sets of
\emph{attributes.} We tacitly identify attributes with their names, i.e.,
attributes are considered as ``names of columns'' in data tables.
As usual, we assume that each attribute in a relation scheme has its \emph{type}
which defines a (possibly infinite but at most denumerable) set of admissible
values for the attribute. We write $u = v$ whenever two value $u,v$ of a particular
type are indistinguishable and $u \ne v$ otherwise.
\emph{Tuples}, which formalize ``rows in data tables''
are considered as maps assigning to each attribute from relation schemes a value
of its type; we denote by $r(y)$ the \emph{$y$-value} of tuple $r$. Furthermore, we
denote by $\mathrm{Tupl}(R)$ the \emph{set of all tuples} on the relation
scheme $R$. Again, note that $\mathrm{Tupl}(R)$ may be infinite.
For $r_1,r_2 \in \mathrm{Tupl}(R)$, we put $r_1 = r_2$ whenever $r_1(y) = r_2(y)$
for all $y \in R$ and $r_1 \ne r_2$ otherwise.
Tuples $r \in \mathrm{Tupl}(R)$ and $s \in \mathrm{Tupl}(S)$ are
called \emph{joinable} whenever $r(y) = s(y)$ for all $y \in R \cap S$.
If $r \in \mathrm{Tupl}(R)$ and $s \in \mathrm{Tupl}(S)$ are joinable,
then $rs$, called the \emph{join} of $r$ and $s$, is a tuple
in $\mathrm{Tupl}(R \cup S)$ such that $(rs)(y) = r(y)$
for $y \in R$ and $(rs)(y) = s(y)$ for $y \in S$.

\begin{remark}\label{rem:emptup}%
  Note that the join of tuples is also called a concatenation and it may be seen
  as a set-theoretic union of tuples since tuples are considered as sets of
  attribute-value pairs, see~\cite{Mai:TRD}. In a special case for $R = \emptyset$
  (the empty relation scheme), $\mathrm{Tupl}(R)$ consists of a single
  tuple---the empty tuple which, according to the set-theoretic representation
  of tuples, may be identified with the empty set.
  Thus, we write $\mathrm{Tupl}(\emptyset) = \{\emptyset\}$.
\end{remark}

Let $\mathbf{L} = \langle L,\leq,0,1\rangle$ be a totally ordered complete lattice.
The elements in $L$ are called \emph{scores}. The scores have a comparative meaning.
That is, if $a < b$ for $a,b \in L$, then $b$ is a score of a better match than $a$.
As a consequence, $1$ is the score of a best match (a full match) and $0$ 
is the score of a worst match (no match).
Considering $\mathbf{L}$ as the structure of scores,
a \emph{ranked data table} (shortly, an RDT) $\mathcal{D}$ on relation scheme $R$
which uses scores in $\mathbf{L}$ is understood as a map
\begin{align}
  \mathcal{D}\!: \mathrm{Tupl}(R) \to L
  \label{eqn:D}
\end{align}
such that $\{r \in \mathrm{Tupl}(R); \mathcal{D}(r) > 0\}$,
called the \emph{answer set} of $\mathcal{D}$, is finite.
That is, only finitely many tuples in $\mathrm{Tupl}(R)$ are assigned
non-zero scores by an RDT $\mathcal{D}$ on $R$;
$\mathcal{D}(r)$ is the score of tuple $r$ in RDT $\mathcal{D}$.
RDTs defined on non-empty relation schemes
may be represented by two-dimensional tables with rows corresponding to
tuples from the answer set, columns corresponding to attributes, and an
extra column (denoted by~$\texttt{\rank}$) containing the scores. For illustration,
if $L = [0,1]$ and $\leq$ is the usual order of reals, the tables in
Fig.\,\ref{fig:RDTs} may be viewed as ranked data tables with scores
in $L = [0,1]$. In the figure, the tuples from the answer set are
sorted according to their scores.

\begin{figure}
  \centering
  \begin{relation}{3}
    \rank &
    \atr{ID} &
    \atr{BDRM} &
    \atr{SQFT} \\
    \hline
    \rule{0pt}{9pt}%
    $1.000$ & \val{85} & \val{5} & \val{4580} \\
    $0.971$ & \val{56} & \val{3} & \val{3400} \\
    $0.937$ & \val{71} & \val{3} & \val{3280} \\
    $0.643$ & \val{82} & \val{4} & \val{2350} \\
    $0.426$ & \val{58} & \val{4} & \val{1760} \\
    $0.148$ & \val{93} & \val{2} & \val{1130} \\
  \end{relation}
  \quad
  \begin{relation}{3}
    \rank &
    \atr{ID} &
    \atr{AGENT} &
    \atr{PRICE} \\
    \hline
    \rule{0pt}{9pt}%
    $0.997$ & \val{71} & \val{Black} & \val{\usd{798,000}} \\
    $0.964$ & \val{58} & \val{Black} & \val{\usd{829,000}} \\
    $0.940$ & \val{71} & \val{Adams} & \val{\usd{849,000}} \\
    $0.798$ & \val{45} & \val{Adams} & \val{\usd{654,000}} \\
    $0.789$ & \val{82} & \val{Adams} & \val{\usd{648,000}} \\
    $0.778$ & \val{85} & \val{Black} & \val{\usd{998,000}} \\
    $0.708$ & \val{45} & \val{Black} & \val{\usd{598,000}} \\
    $0.708$ & \val{93} & \val{Black} & \val{\usd{598,000}} \\
  \end{relation}
  \caption{%
    Examples of ranked tables:
    Houses with area roughly greater than \val{3,500} sq. ft. (left),
    houses offered by agents for about \usd{\$800,000} (right).}
  \label{fig:RDTs}
\end{figure}

In the borderline case of $R = \emptyset$, the answer set of $\mathcal{D}$ on $R$
contains at most the empty tuple $\emptyset$. If the answer set is empty then
clearly $\mathcal{D}(\emptyset) = 0$. Otherwise, $\mathcal{D}$ is uniquely
given by the non-zero score $\mathcal{D}(\emptyset) \in L$. Note that this
naturally generalizes the two borderline relations on the empty relation scheme $R$
which appear in the classic model: the empty relation on $R$ and the relation
on $R$ containing the empty tuple.

Furthermore, we consider equality of RDTs as follows:
For RDTs $\mathcal{D}_1$ and $\mathcal{D}_2$ on $R$ we put 
$\mathcal{D}_1 = \mathcal{D}_2$ whenever
$\mathcal{D}_1(r) = \mathcal{D}_2(r)$ for all $r \in \mathrm{Tupl}(R)$,
i.e., whenever $\mathcal{D}_1$ and $\mathcal{D}_2$ are equal as maps.
The \emph{range} (\emph{or scores}) of RDT $\mathcal{D}$ on relation scheme $R$,
denoted $L(\mathcal{D})$, is a subset of $L$ defined by
\begin{align}
  L(\mathcal{D}) &= \{\mathcal{D}(r);\, r \in \mathrm{Tupl}(R)\}.
  \label{eqn:LD}
\end{align}
That is, $L(\mathcal{D})$ is the set of all scores from $L$ which appear
in $\mathcal{D}$. Thus, $L(\mathcal{D})$ is finite for any $\mathcal{D}$.
Let us note here that if $L(\mathcal{D}) \subseteq \{0,1\}$,
then $\mathcal{D}$ may be viewed as a ranked representation of a classic
relation on a relation scheme. Indeed, for an ordinary (finite)
$\mathcal{R} \subseteq \mathrm{Tupl}(R)$, we can introduce a
corresponding RDT $\mathcal{D}_\mathcal{R}$ by putting
$\mathcal{D}_\mathcal{R}(r) = 1$ whenever $r \in \mathcal{R}$
and $\mathcal{D}_\mathcal{R}(r) = 0$ otherwise.
Conversely, for RDT $\mathcal{D}$ on $R$, we may consider
a corresponding $\mathcal{R}_\mathcal{D} \subseteq \mathrm{Tupl}(R)$
as $\mathcal{R}_\mathcal{D} = \{r;\, \mathcal{D}(r) = 1\}$.
Taking into account just RDTs with ranges being subsets of $\{0,1\}$,
the two transformations are mutually inverse. As a consequence,
in the same spirit as in the Codd model~\cite{Co:Armodflsdb},
RDTs may represent both the results of queries and base data, i.e.,
our approach uses only a single type of structures. As a consequence,
we do not mix the classic relations and the ranked data tables.

\begin{remark}\label{rem:TTM_repr}%
  (a)
  The fundamental notion of a ranked data table may seem like a digression
  from the relational model of data and in particular from its modern
  understanding as it is described in \emph{The Third Manifesto} (\emph{TTM},
  see~\cite{DaDa:DTRM}) because tuples in relations are annotated by an
  additional information which is the score. If one wishes the approach
  to adhere to \emph{TTM}, he can consider the score as an additional
  attribute (named $\rank$) which is present in the relation scheme.
  The type of the attribute $\rank$ is \emph{score}. In other words,
  RDTs may be seen as ordinary relations on relation schemes with
  a special designated attribute $\rank$ the values of which come
  from the universe of $\mathbf{L}$.

  (b)
  As we have mentioned in the introduction, our model is not related to
  probabilistic databases which are currently extensively studied. In particular,
  the scores \emph{cannot} be interpreted as probabilities. Let us note that the
  scores need not come from a real unit interval, so in general it does not make sense
  to consider the scores as probability values. Even if the scores do come from
  a unit interval, their values are not related to probabilities assigned to
  any events because there is no uncertainty involved in the data or in
  query evaluation as we shall se later.

  (c)
  Note that various approaches where tuples in relations are annotated by
  values coming from general algebraic structures exist. Most notably,
  the authors of \cite{ImLi:Iird} consider conditional tables which may be
  understood as relations with tuples annotated by Boolean formulas, i.e.,
  annotated by
  values coming from particular free Boolean algebras. A general approach
  to relations annotated by element from semi-rings is presented
  in~\cite{GrKaTa:PS}, see also \cite{FoGrTa:AXQP}, \cite{Gr:CCQAR},
  and \cite{AmDeTa:PAQ}.
\end{remark}

We now describe a set of relational operations which are used to express queries
over ranked data tables. Important types of monotone as well as non-monotone
queries in rank-aware databases may be expressed by a combination of the
following operations with RDTs which generalize their classic relational
operations in the original relational model of data. For the introduced operations,
we adopt the widely used Codd-style notation.

Let $\mathcal{D}_1$ and $\mathcal{D}_2$ be RDTs on $R \cup S$ and $S \cup T$
with $R \cap S = \emptyset$, $R \cap T = \emptyset$, and $S \cap T = \emptyset$.
The (\emph{natural}) \emph{join} of $\mathcal{D}_1$ and $\mathcal{D}_2$,
denoted $\join{\mathcal{D}_1}{\mathcal{D}_2}$, is defined by
\begin{align}
  (\join{\mathcal{D}_1}{\mathcal{D}_2})(rst) &=
  \inf\{\mathcal{D}_1(rs),\mathcal{D}_2(st)\},
  \label{eqn:join}
\end{align}
for all $r \in \mathrm{Tupl}(R)$, $s \in \mathrm{Tupl}(S)$,
and $t \in \mathrm{Tupl}(T)$. Recall that $rs$, $st$ and $rst$
in~\eqref{eqn:join} denote the results of joins of tuples which
are in this case trivially joinable. Since $\mathbf{L}$ is totally
ordered, the score $(\join{\mathcal{D}_1}{\mathcal{D}_2})(rst)$ in
\eqref{eqn:join} is in fact taken as the minimum of the scores
$\mathcal{D}_1(rs)$ and $\mathcal{D}_2(st)$. Also note that the
commutativity, associativity, and idempotency of $\inf$ implies
that $\bowtie$ has these properties as well. In addition, any $\mathcal{D}$
(over any $R$) with an empty answer set is an annihilator with respect
to $\bowtie$ and $\mathcal{D}$ on $\emptyset$ such that
$\mathcal{D}(\emptyset) = 1$ is a neutral element with respect to $\bowtie$.
The join of the illustrative RDTs in Fig.\,\ref{fig:RDTs} is shown in
Fig.\,\ref{fig:join}.

\begin{remark}\label{rem:general}
  As we have noted in the introduction, the operations with RDTs we use in this paper
  may be viewed as particular cases of those used in~\cite{BeVy:Qssbd}.
  In~\cite{BeVy:Qssbd}, the basic structures of scores are
  complete residuated lattices~\cite{GaJiKoOn:RL} which may be viewed as
  generalization of the structures of scores defined on the real unit interval
  by left-continuous triangular norms~\cite{EsGo:MTL}. The general counterpart
  to~\eqref{eqn:join} in~\cite{BeVy:Qssbd} is $\bowtie_\otimes$ defined by
  \begin{align}
    (\mathcal{D}_1 \bowtie_\otimes \mathcal{D}_2)(rst) &=
    \mathcal{D}_1(rs) \otimes \mathcal{D}_2(st)
    \label{eqn:join*}
  \end{align}
  for all $r \in \mathrm{Tupl}(R)$, $s \in \mathrm{Tupl}(S)$,
  and $t \in \mathrm{Tupl}(T)$. Obviously, \eqref{eqn:join} is a particular
  case of~\eqref{eqn:join*} with $\otimes$ being $\inf$, i.e.,
  $a \otimes b = \inf\{a,b\}$ for all $a,b \in L$. Further in the paper
  we show that joins defined by~\eqref{eqn:join} are invariant to ordinal
  transformations provided that totally ordered G\"odel algebras are used
  as structures of scores.
  We also show that the property does not hold in the general setting
  of complete residuated lattices. A similar remark can be made for all
  the operations introduced below.
\end{remark}

\begin{figure}
  \centering
  \begin{relation}{5}
    \rank &
    \atr{ID} &
    \atr{BDRM} &
    \atr{SQFT} &
    \atr{AGENT} &
    \atr{PRICE} \\
    \hline
    \rule{0pt}{9pt}%
    $0.937$ & \val{71} & \val{3} & \val{3280} & \val{Adams} & \val{\usd{849,000}} \\
    $0.937$ & \val{71} & \val{3} & \val{3280} & \val{Black} & \val{\usd{798,000}} \\
    $0.778$ & \val{85} & \val{5} & \val{4580} & \val{Black} & \val{\usd{998,000}} \\
    $0.643$ & \val{82} & \val{4} & \val{2350} & \val{Adams} & \val{\usd{648,000}} \\
    $0.426$ & \val{58} & \val{4} & \val{1760} & \val{Black} & \val{\usd{829,000}} \\
    $0.148$ & \val{93} & \val{2} & \val{1130} & \val{Black} & \val{\usd{598,000}} \\
  \end{relation}
  \caption{Join of ranked tables from Figure~\ref{fig:RDTs}.}
  \label{fig:join}
\end{figure}

\begin{figure}
  \centering
  \begin{relation}{5}
    \rank &
    \atr{ID} &
    \atr{BDRM} &
    \atr{SQFT} &
    \atr{AGENT} &
    \atr{PRICE} \\
    \hline
    \rule{0pt}{9pt}%
    $0.939$ & \val{71} & \val{3} & \val{3280} & \val{Black} & \val{\usd{798,000}} \\
    $0.938$ & \val{71} & \val{3} & \val{3280} & \val{Adams} & \val{\usd{849,000}} \\
    $0.778$ & \val{85} & \val{5} & \val{4580} & \val{Black} & \val{\usd{998,000}} \\
    $0.643$ & \val{82} & \val{4} & \val{2350} & \val{Adams} & \val{\usd{648,000}} \\
    $0.426$ & \val{58} & \val{4} & \val{1760} & \val{Black} & \val{\usd{829,000}} \\
    $0.148$ & \val{93} & \val{2} & \val{1130} & \val{Black} & \val{\usd{598,000}} \\
  \end{relation}
  \caption{RDT ``similar'' to that in Figure~\ref{fig:join}.}
  \label{fig:join_simil}
\end{figure}

We introduce restrictions (selections) of RDTs utilizing general maps
serving as restriction conditions: By a \emph{restriction condition} on $R$
we mean any map $\rcond\!: \mathrm{Tupl}(R) \to L$ with each $\rcond(r)$
interpreted as the score expressing whether (and to what degree)
tuple $r$ matches $\rcond$. Note that the ordinary restriction conditions
based on classic comparators of domain values are covered by this general
notion. For instance, if $\rcond\!: \mathrm{Tupl}(R) \to L$ is defined
so that $\rcond(r) = 1$ whenever $r(y_1) = r(y_2)$ and $\rcond(r) = 0$ otherwise,
then $\rcond$ may be seen as representing a classic restriction condition based
on equality of the values of attributes $y_1$ and $y_2$.

Given RDT $\mathcal{D}$ on $R$, we define the \emph{restriction} of
$\mathcal{D}$ using $\rcond$ on $R$ by
\begin{align}
  (\restrict{\rcond}{\mathcal{D}})(r) &=
  \inf\{\mathcal{D}(r),\rcond(r)\}
  \label{eqn:restrict}
\end{align}
for all $r \in \mathrm{Tupl}(R)$.
Obviously, the score of $r$ in $\restrict{\rcond}{\mathcal{D}}$ at most
as high as its score in $\mathcal{D}$ which is a natural property of a restriction.
Note that the ranked tables in Fig.~\ref{fig:RDTs} may be seen as results of
particular restrictions of base data tables with all scores (of tuples present
in the tables) set to $1$.

For $\mathcal{D}$ on $R$ and $S \subseteq R$, the \emph{projection}
$\project{S}{\mathcal{D}}$ of $\mathcal{D}$ onto $S$ is defined by
\begin{align}
  (\project{S}{\mathcal{D}})(s) &=
  \sup\{\mathcal{D}(st);\, t \in \mathrm{Tupl}(R \setminus S)\}
  \label{eqn:project}
\end{align}
for all $s \in \mathrm{Tupl}(S)$. Here, notice the use of $\sup$ instead of
$\inf$ which corresponds to the close relationship of projections and
existentially quantified queries. Recall that in the classic setting,
the fact that $s$ belongs to a projection of a relation onto $S$ means
that \emph{there exists} $t$ such that $st$ is in the relation. In a similar
sense, the score of $s$ in $\project{S}{\mathcal{D}}$ is defined by
\eqref{eqn:project} as the highest score of $st$ over all $t$ (note that
two different tuples in the answer set of $\mathcal{D}$ may be projected
onto the same tuple on $S$).

For $\mathcal{D}_1$ and $\mathcal{D}_2$ on the same relation scheme $R$,
we introduce the \emph{union} of $\mathcal{D}_1$ and
$\mathcal{D}_2$ which is defined componentwise using $\sup$ as
\begin{align}
  (\union{\mathcal{D}_1}{\mathcal{D}_2})(r) &=
  \sup\{\mathcal{D}_1(r),\mathcal{D}_2(r)\},
  \label{eqn:union}
\end{align}
for all $r \in \mathrm{Tupl}(R)$. In addition, we may consider an intersection
based on $\inf$ but this operation is superfluous because it can be understood
as a join~\eqref{eqn:join} of two RDTs on the same relation scheme.

Since $\mathbf{L}$ is linearly ordered, $\diff$ which is adjoint to $\sup$
as in~\eqref{eqn:dadj} exists and it is given by~\eqref{eqn:diff_chain}.
Therefore, for $\mathcal{D}_1$ and $\mathcal{D}_2$ on the same relation
scheme $R$, we may introduce the \emph{difference} of $\mathcal{D}_1$ and
$\mathcal{D}_2$ by
\begin{align}
  (\minus{\mathcal{D}_1}{\mathcal{D}_2})(r) &=
  \mathcal{D}_1(r) \diff \mathcal{D}_2(r),
\end{align}
for all $r \in \mathrm{Tupl}(R)$, i.e.,
using~\eqref{eqn:diff_chain} and the total ordering of $\mathbf{L}$,
\begin{align}
  (\minus{\mathcal{D}_1}{\mathcal{D}_2})(r) &=
  \left\{
    \begin{array}{@{\,}l@{\quad}l@{}}
      0, &\text{if } \mathcal{D}_1(r) \leq \mathcal{D}_2(r), \\[2pt]
      \mathcal{D}_1(r), &\text{otherwise,}
    \end{array}
  \right.
  \label{eqn:minus}
\end{align}
for all $r \in \mathrm{Tupl}(R)$.

Finally, we consider operations with RDTs which are related to universally
quantified queries. In the classic model, queries of the form of categorical
propositions ``every $\varphi$ is $\psi$'' may be expressed by divisions
(or more general constructs such as the imaging operator considered
in~\cite{DaDa:imrel}) which are in the classic model expressible by
means of other operations (joins, projections, and difference). From the logical
point of view this is a consequence of the fact that the universal quantifier
is definable using negations and the existential quantifier. As we shall see
in Section~\ref{sec:heyting}, this property does not hold in a weaker logic
which is closely related to the rank-aware model. Therefore, in our case,
we have to introduce an operation in order to be able to properly express
queries of the form of categorical propositions ``every $\varphi$ is $\psi$''.
In our case, such an operation will be a variant of the Small Divide as
it is considered in~\cite{DaDa:divop}.

Note that analogously as in the case of $\diff$, the residuum $\rightarrow$
satisfying~\eqref{eqn:adj} always exists and is uniquely given
by~\eqref{eqn:resd_chain} owing to the linearity of~$\mathbf{L}$.
Let $\mathcal{D}_1$ (so-called mediator) be an RDT on $R \cup S$ such
that $R \cap S = \emptyset$, $\mathcal{D}_2$ (so-called divisor) be an RDT on $S$,
and let $\mathcal{D}_3$ (so-called dividend) be an RDT on $R$.
In this setting, we introduce a \emph{division}
$\rdiv{\mathcal{D}_3}{\mathcal{D}_1}{\mathcal{D}_2}$ as an RDT on $R$ such that
\begin{align}
  \bigl(\rdiv{\mathcal{D}_3}{\mathcal{D}_1}{\mathcal{D}_2}\bigr)(r) =
  \textstyle
  \inf\{\mathcal{D}_2(s) \rightarrow  \mathcal{D}_1(rs);\,
  s \in \mathrm{Tupl}(S)
  \} \cup \{\mathcal{D}_3(r)\},
  \label{eqn:div}
\end{align}
for all $r \in \mathrm{Tupl}(R)$. Directly from~\eqref{eqn:div}, the answer set
of $\rdiv{\mathcal{D}_3}{\mathcal{D}_1}{\mathcal{D}_2}$ is finite since it is
a subset of the answer set of $\mathcal{D}_3$. By moment's reflection, we can see
that~\eqref{eqn:div} can equivalently be written as
\begin{align}
  \bigl(\rdiv{\mathcal{D}_3}{\mathcal{D}_1}{\mathcal{D}_2}\bigr)(r) =
  \textstyle
  \inf\{\inf\{\mathcal{D}_3(r), \mathcal{D}_2(s) \rightarrow  \mathcal{D}_1(rs)\};\,
  s \in \mathrm{Tupl}(S)\}.
  \label{eqn:div_alt}
\end{align}
Observe that according to~\eqref{eqn:div_alt} and~\eqref{eqn:resd_chain},
$\bigl(\rdiv{\mathcal{D}_3}{\mathcal{D}_1}{\mathcal{D}_2}\bigr)(r) = \mathcal{D}_3(r)$
if{}f
\begin{align}
  \text{for all } s \in \mathrm{Tupl}(S)\!:\,
  \mathcal{D}_2(s) > \mathcal{D}_1(rs) \text{ implies }
  \mathcal{D}_3(r) \leq \mathcal{D}_1(rs).
  \label{eqn:div_cond}
\end{align}
Therefore, we can distinguish two cases as follows:
\begin{align}
  \bigl(\rdiv{\mathcal{D}_3}{\mathcal{D}_1}{\mathcal{D}_2}\bigr)(r) &=
  \left\{
    \begin{array}{@{\,}l@{\quad}l@{}}
      \mathcal{D}_3(r), &\text{if } \eqref{eqn:div_cond} \text{ holds,}
      \\[2pt]
      \textstyle
      \inf\{\mathcal{D}_2(s) \rightarrow \mathcal{D}_1(rs);\,
      s \in \mathrm{Tupl}(S)
      \}, 
      &\text{otherwise.}
    \end{array}
  \right.
\end{align}
Using~\eqref{eqn:resd_chain} again, we have
\begin{align}
  \bigl(\rdiv{\mathcal{D}_3}{\mathcal{D}_1}{\mathcal{D}_2}\bigr)(r) &=
  \left\{
    \begin{array}{@{\,}l@{\quad}l@{}}
      \mathcal{D}_3(r), &\text{if } \eqref{eqn:div_cond} \text{ holds,}
      \\[2pt]
      \textstyle
      \inf\{\mathcal{D}_1(rs);\,
      \mathcal{D}_2(s) > \mathcal{D}_1(rs)
      \text{, } s \in \mathrm{Tupl}(S)
      \}, &\text{otherwise.}
    \end{array}
  \right.
  \label{eqn:div_expr}
\end{align}
As a consequence, the rank of a tuple in the result of a division can always be
computed in finitely many steps because each divisor has a finite answer set.

Closely related to the division is the notion of a subsethood (inclusion of RDTs)
which, in our case, can also be expressed by a score.
Namely, for RDTs $\mathcal{D}_1$ and $\mathcal{D}_2$ on the same relation scheme $R$,
we put
\begin{align}
  \gsubs{\mathcal{D}_1}{\mathcal{D}_2} &=
  \inf\{\mathcal{D}_1(r) \rightarrow \mathcal{D}_2(r);\, r \in \mathrm{Tupl}(R)\}
  \label{eqn:gsubs_general}
  \\
  &= 
  \inf\{\mathcal{D}_2(r);\, \mathcal{D}_1(r) > \mathcal{D}_2(r) \text{ and }
  r \in \mathrm{Tupl}(R)\}
  \label{eqn:gsubs}
\end{align}
and call $\gsubs{\mathcal{D}_1}{\mathcal{D}_2}$
the \emph{subsethood score} of $\mathcal{D}_1$ in $\mathcal{D}_2$.
That is, $\gsubsop$ is not a relational operation because its result is
a score in $\mathbf{L}$ (and not an RDT).
The subsethood scores generalize the concept of containment of relations.
Indeed, if ranked tables $\mathcal{D}_1$ and $\mathcal{D}_2$ are considered
as results
of queries $Q_1$ and $Q_2$, then $\gsubs{\mathcal{D}_1}{\mathcal{D}_2}$
is the score expressing the degree to which ``if a tuple satisfies $Q_1$,
then it satisfies $Q_2$'' is satisfied by all tuples.
In particular, it is easily seen that $\mathcal{D}_1 = \mathcal{D}_2$ if{}f
$\gsubs{\mathcal{D}_1}{\mathcal{D}_2} =
\gsubs{\mathcal{D}_2}{\mathcal{D}_1} = 1$.
Subsethood scores are related to division as follows:
For $R = \emptyset$ and $\mathcal{D}_1$ and $\mathcal{D}_2$ being RDTs on $S$,
we get that $\gsubs{\mathcal{D}_1}{\mathcal{D}_2} =
(\rdiv{\mathcal{D}}{\mathcal{D}_2}{\mathcal{D}_1})(\emptyset)$,
where $\mathcal{D}$ is the RDT on $\emptyset$ such that $\mathcal{D}(\emptyset) = 1$.

\begin{example}
  If we consider the RDTs in Fig.\,\ref{fig:join} and Fig.\,\ref{fig:join_simil}
  and denote them as $\mathcal{D}_1$ and $\mathcal{D}_2$, respectively,
  then $\gsubs{\mathcal{D}_1}{\mathcal{D}_2} = 1$ because all scores
  in $\mathcal{D}_1$ are lower than or equal to the scores of the
  corresponding tuples in $\mathcal{D}_2$ (note that here we use the
  fact that $\inf \emptyset = 1$), i.e., we may say that,
  taking the scores into account,
  $\mathcal{D}_1$ is \emph{fully included} in $\mathcal{D}_2$. On the contrary, 
  $\gsubs{\mathcal{D}_2}{\mathcal{D}_1} < 1$.
  Namely, $\gsubs{\mathcal{D}_2}{\mathcal{D}_1} = \inf\{0.937\} = 0.937$.
  Analogously as the subsethood scores, we may consider
  a related notion of a \emph{similarity score}
  $\gequs{\mathcal{D}_1}{\mathcal{D}_2}$ of $\mathcal{D}_1$ and
  $\mathcal{D}_2$ defined as
  \begin{align}
    \gequs{\mathcal{D}_1}{\mathcal{D}_2} &=
    \inf\{\gsubs{\mathcal{D}_1}{\mathcal{D}_2},
    \gsubs{\mathcal{D}_2}{\mathcal{D}_1}\}.
    \label{eqn:gequs}
  \end{align}
  In this case, 
  $\gequs{\mathcal{D}_2}{\mathcal{D}_1} = 0.937$.
\end{example}

Let us note that we can introduce a ternary operation with RDTs which is defined
componentwise using $\rightarrow$ in a similar way as the union of RDTs which is
defined componentwise using $\sup$: For $\mathcal{D}_1$, $\mathcal{D}_2$,
and $\mathcal{D}_3$ on the same relation scheme $R$, we put
\begin{align}
  \bigl(\resd{\mathcal{D}_3}{\mathcal{D}_1}{\mathcal{D}_2}\bigr)(r) &=
  \inf\{\mathcal{D}_3(r), \mathcal{D}_1(r) \rightarrow \mathcal{D}_2(r)\}
  \label{eqn:resd_op}
\end{align}
for all $r \in \mathrm{Tupl}(R)$ and
call $\resd{\mathcal{D}_3}{\mathcal{D}_1}{\mathcal{D}_2}$
the \emph{$\mathcal{D}_3$-residuum} of $\mathcal{D}_1$
with respect to $\mathcal{D}_2$. Note that the operation is correct
in that the result is always an RDT, i.e., there are only finitely
many tuples for which~\eqref{eqn:resd_op} is non-zero. Analogously as
in the case of $\frdiv$, we get
\begin{align}
  \bigl(\resd{\mathcal{D}_3}{\mathcal{D}_1}{\mathcal{D}_2}\bigr)(r) &=
  \left\{
    \begin{array}{@{\,}l@{\quad}l@{}}
      \mathcal{D}_3(r), &\text{if } \mathcal{D}_1(r) \leq \mathcal{D}_2(r)
      \text{ or }
      \mathcal{D}_3(r) \leq \mathcal{D}_2(r),
      \\[2pt]
      \mathcal{D}_2(r), &\text{otherwise,}
    \end{array}
  \right.
  \label{eqn:resd_cases}
\end{align}
which follows easily by~\eqref{eqn:resd_chain}. Since $\rightarrow$ acts
in a similar way as the truth function of the classic implication,
\eqref{eqn:resd_cases} may be seen as expressing the score of
a condition ``$r$ belongs to $\mathcal{D}_3$ and if it belongs to
$\mathcal{D}_1$, then it belongs to $\mathcal{D}_2$''.

\begin{remark}
  (a)
  The operations of join, projection, union, difference, and division behave
  the same way as their ordinary counterparts when the scores in the input RDTs
  are only $0$ and $1$ (i.e., their range is a subset of $\{0,1\}$). In addition,
  the restriction also behaves as the ordinary restriction provided that the
  input RDT has only scores $0$ and $1$ and that the range of the 
  restriction condition is also a subset of $\{0,1\}$. In general, restrictions
  produce RDTs with general scores: The tables in Fig.~\ref{fig:RDTs} may be seen
  as such examples.

  (b)
  Form the point of view of the representation of RDTs as ordinary relations with
  a~special attribute $\rank$, see Remark~\ref{rem:TTM_repr}\,(a), we may think of
  the operations introduced in this section as derived operations which always
  produce a relation with $\rank$ (representing the output RDT) from other relations
  with $\rank$ (representing the input RDTs). From the perspective of \emph{TTM} and
  in particular the relational query language \emph{Tutorial D}, the operations may
  be implemented as user defined operators in a similar fashion as the operators
  supporting operations with temporal data described in~\cite{DDL:TRT}.
\end{remark}

\section{Ordinal Equivalence of Tables}\label{sec:ordeq}
We introduce notions of ordinal inclusion and equivalence of ranked data tables
based on positions of tuples in tables given by scores. In the next section,
we utilize the notion in a characterization of important order-related properties
of the relational operations with RDTs. Intuitively, we may consider
$\mathcal{D}_1$ and $\mathcal{D}_2$ (on the same relation scheme) ordinally
equivalent if the sequences of tuples in $\mathcal{D}_1$ and $\mathcal{D}_2$
sorted by scores are identical. Formally, we introduce the notion
as follows.

\begin{definition}
  For any $\mathcal{D}$ on $R$ and $r \in \mathrm{Tupl}(R)$, we put
  \begin{align}
    \mathcal{U}(\mathcal{D},r) &=
    \{r' \in \mathrm{Tupl}(R);\, \mathcal{D}(r') \geq \mathcal{D}(r)\}.
    \label{eqn:U}
  \end{align}
  For any $\mathcal{D}_1$ and $\mathcal{D}_2$ on $R$ we say that
  $\mathcal{D}_1$ is \emph{ordinally included} in $\mathcal{D}_2$,
  written $\mathcal{D}_1 \osub \mathcal{D}_2$, whenever
  \begin{align}
    \mathcal{U}(\mathcal{D}_1,r) \subseteq \mathcal{U}(\mathcal{D}_2,r)
    \label{eqn:osub}
  \end{align}
  for all $r \in \mathrm{Tupl}(R)$.
  Moreover, we call $\mathcal{D}_1$ and $\mathcal{D}_2$
  \emph{ordinally equivalent,}
  written $\mathcal{D}_1 \oeqv \mathcal{D}_2$, whenever
  \begin{align}
    \mathcal{D}_1 \osub \mathcal{D}_2 \text{ and }
    \mathcal{D}_2 \osub \mathcal{D}_1.
    \label{eqn:oeqv}
  \end{align}
\end{definition}

\begin{remark}
  We can immediately observe properties of $\mathcal{U}$, $\osub$,
  and $\oeqv$ which follow directly by the definition: First,
  $r \in \mathcal{U}(\mathcal{D},r)$ follows by the reflexivity of $\leq$.
  Second, if $\mathcal{D}(r) = 0$, then $\mathcal{U}(\mathcal{D},r) = \mathrm{Tupl}(R)$
  and it is infinite if $R$ contains an attribute of an infinite type. If
  $\mathcal{D}(r) > 0$, then $\mathcal{U}(\mathcal{D},r)$ is always finite and
  it is a subset of the answer set of $\mathcal{D}$ which follows directly
  by~\eqref{eqn:U}.
  Third, $\osub$ is reflexive (a consequence of the reflexivity of $\subseteq$)
  and transitive (a consequence of the transitivity of $\subseteq$) and thus
  $\osub$ is a preorder (also called a quasi order~\cite{Bir:LT}).
  In general, $\osub$ is not a partial order because it is not antisymmetric.
  Indeed, consider the RDTs $\mathcal{D}_1$ and $\mathcal{D}_2$ on $R = \{\atr{foo}\}$
  in Fig.~\ref{fig:not_antisym}. For the only tuple $r$ which appears in the answer
  set of both the RDTs, we have $\mathcal{U}(\mathcal{D}_1,r) = \{r\} =
  \mathcal{U}(\mathcal{D}_2,r)$ which from it readily follows that
  $\mathcal{D}_1 \osub \mathcal{D}_2$ and $\mathcal{D}_2 \osub \mathcal{D}_1$.
  Fourth, in general, $\osub$ has no relationship to the inclusion of answer sets.
  For instance, if $\mathcal{D}_2$ has an empty answer set
  (i.e., $\mathcal{D}_2(r) = 0$ for all $r \in \mathrm{Tupl}(R)$),
  then trivially $\mathcal{D}_1 \osub \mathcal{D}_2$ for any $\mathcal{D}_1$ on $R$.
  Fifth, by definition, $\oeqv$ is the symmetric interior of $\osub$ (i.e., the
  greatest symmetric relation contained in both $\osub$ and its inverse)
  and therefore it is an equivalence relation.
\end{remark}

Dually to $\mathcal{U}(\mathcal{D},r)$,
we may introduce $\mathcal{L}(\mathcal{D},r)$ by
\begin{align}
  \mathcal{L}(\mathcal{D},r) &=
  \{r' \in \mathrm{Tupl}(R);\, \mathcal{D}(r') \leq \mathcal{D}(r)\}
  \label{eqn:L}
\end{align}
for any $r \in \mathrm{Tupl}(R)$. Therefore, in contrast to $\mathcal{U}(\mathcal{D},r)$
which represents the set of tuples in $\mathcal{D}$ which have scores
at least as high as $\mathcal{D}(r)$, $\mathcal{L}(\mathcal{D},r)$ is
the set of tuples with scores at most as high as $\mathcal{D}(r)$.
It is easy to see that $\osub$ and $\oeqv$ can equivalently be defined
using \eqref{eqn:L} instead of~\eqref{eqn:U} which is justified by the following
assertion.

\begin{theorem}\label{th:L}%
  $\mathcal{D}_1 \osub \mathcal{D}_2$ if{}f
  for all $r \in \mathrm{Tupl}(R)$, we have
  $\mathcal{L}(\mathcal{D}_1,r) \subseteq \mathcal{L}(\mathcal{D}_2,r)$.
\end{theorem}
\begin{proof}
  Let $\mathcal{D}_1 \osub \mathcal{D}_2$ and consider any
  $r \in \mathrm{Tupl}(R)$. Furthermore, let $r' \in \mathcal{L}(\mathcal{D}_1,r)$,
  i.e., $\mathcal{D}(r') \leq \mathcal{D}(r)$ by~\eqref{eqn:L}.
  Using~\eqref{eqn:U}, the last inequality gives
  $r \in \mathcal{U}(\mathcal{D}_1,r')$ and thus 
  $r \in \mathcal{U}(\mathcal{D}_2,r')$ because 
  $\mathcal{D}_1 \osub \mathcal{D}_2$. Now,
  from $r \in \mathcal{U}(\mathcal{D}_2,r')$
  it follows that $r' \in \mathcal{L}(\mathcal{D}_2,r)$.
  As a consequence,
  $\mathcal{L}(\mathcal{D}_1,r) \subseteq \mathcal{L}(\mathcal{D}_2,r)$.
  The converse implication can be shown by analogous arguments utilizing
  the fact that for any $\mathcal{D}$ on $R$ and arbitrary
  tuples $r,r' \in \mathrm{Tupl}(R)$, we have
  $r' \in \mathcal{U}(\mathcal{D},r)$ if{}f
  $r \in \mathcal{L}(\mathcal{D},r')$.
\end{proof}

\begin{figure}
  \centering
  \begin{relation}{1}
    \rank &
    \atr{FOO} \\
    \hline
    \rule{0pt}{9pt}%
    $0.600$ & \val{77} \\
  \end{relation}
  \quad
  \begin{relation}{1}
    \rank &
    \atr{FOO} \\
    \hline
    \rule{0pt}{9pt}%
    $0.500$ & \val{77} \\
  \end{relation}
  \caption{%
    Two distinct RDTs such that
    $\mathcal{D}_1 \osub \mathcal{D}_2$ and $\mathcal{D}_2 \osub \mathcal{D}_1$.}
  \label{fig:not_antisym}
\end{figure}

\begin{example}\label{ex:ord}%
  Recall the RDTs $\mathcal{D}_1$ and $\mathcal{D}_2$ given by
  the tables in Fig.~\ref{fig:join} and Fig.~\ref{fig:join_simil},
  respectively. As one can check, we have $\mathcal{D}_2 \osub \mathcal{D}_1$,
  i.e., $\mathcal{D}_2$ is ordinally included in $\mathcal{D}_1$.
  On the other hand, $\mathcal{D}_1 \nosub \mathcal{D}_2$ because for
  $r \in \mathrm{Tupl}(R)$ such that $r(\atr{price}) = \val{\$798,000}$,
  we have
  \begin{align*}
    \mathcal{U}(\mathcal{D}_1,r) =
    \{r,r'\} \nsubseteq
    \{r\} =
    \mathcal{U}(\mathcal{D}_2,r),
  \end{align*}
   where $r'(\atr{PRICE}) = \val{\$849,000}$. As a consequence,
   $\mathcal{D}_1$ and $\mathcal{D}_2$ are not ordinally equivalent.
\end{example}

The relations of ordinal inclusion and equivalence of ranked tables are closely
related to order-preserving maps and isomorphisms on the structure of scores.
The following definition recalls standard notions of maps between ordered sets
which we use to get insight into the notions of ordinal inclusion and equivalence.

\begin{definition}
  Let $f\!: L_1 \to L_2$ be a map such that $L_1,L_2 \subseteq L$.
  Then, \newline $f$ is called \emph{order preserving} whenever,
  for all $a,b \in L_1$,
  \begin{align}
    a \leq b \text{ implies } f(a) \leq f(b);
    \label{eqn:ord_pres}
  \end{align}
  $f$ is called \emph{order reflecting}
  whenever, for all $a,b \in L_1$,
  \begin{align}
    f(a) \leq f(b) \text{ implies } a \leq b;
    \label{eqn:ord_refl}
  \end{align}
  $f$ is called \emph{order embedding} whenever it is both
  order preserving and order reflecting; \newline
  $f$ is called \emph{order isomorphism}
  whenever it is a surjective order embedding.
\end{definition}

For $\mathcal{D}$ on $R$ and $f\!: L_1 \to L_2$ 
such that $L_1,L_2 \subseteq L$, we may consider
the usual \emph{composition} $\mathcal{D} \circ f$
(written in the diagrammatic notation) defined by
\begin{align}
  (\mathcal{D} \circ f)(r) &= f(\mathcal{D}(r))
  \label{eqn:circ}
\end{align}
for all $r \in \mathrm{Tupl}(R)$; $\mathcal{D} \circ f$ is a correctly
defined map since both $\mathcal{D}$ and $f$ are considered as maps,
see~\eqref{eqn:D}.
Observe that if $f(0) > 0$, then $\mathcal{D} \circ f$ may
not be a ranked table since infinitely many tuples in
$\mathcal{D} \circ f$ may get a non-zero score when $R$ contains
an attribute of an infinite type (we tacitly ignore the fact in the
rest of the paper because it is not relevant to our investigation).
On the other hand, if $f(0) = 0$, then there are only finitely
many $r \in \mathrm{Tupl}(R)$ such that
$(\mathcal{D} \circ f)(r) > 0$, i.e., $\mathcal{D} \circ f$
is always an RDT. In fact, in this case the answer set of
$\mathcal{D} \circ f$ is a subset of the answer set of $\mathcal{D}$.
In addition, if $f$ is order reflecting then it is easily seen that the
answer sets of $\mathcal{D} \circ f$ and $\mathcal{D}$ coincide:
$f(\mathcal{D}(r)) = 0$ yields $f(\mathcal{D}(r)) \leq f(0)$, i.e.,
$\mathcal{D}(r) \leq 0$ by~\eqref{eqn:ord_refl}, meaning that 
if $r$ is in the answer set of $\mathcal{D}$, then it is in the answer
set of $\mathcal{D} \circ f$.

The basic relationship of ordinal inclusion and equivalence
relations and particular order-preserving maps
is described by the following two assertions.

\begin{theorem}\label{th:ord_sub}%
  Let $\mathcal{D}_1$ and $\mathcal{D}_2$ be RDTs on $R$.
  Then, $\mathcal{D}_1 \osub \mathcal{D}_2$ if{}f there is
  an order-preserving map $f\!: L \to L$ such that
  $\mathcal{D}_1 \circ f = \mathcal{D}_2$.
\end{theorem}
\begin{proof}
  In order to prove the only-if part of the assertion,
  assume that $\mathcal{D}_1 \osub \mathcal{D}_2$ and consider
  $f\!: L \to L$ defined by
  \begin{align}
    f(a) &= \textstyle\bigwedge\{\mathcal{D}_2(r);\,
    r \in \mathrm{Tupl}(R) \text{ and } \mathcal{D}_1(r) \geq a\}
    \label{eqn:fa}
  \end{align}
  for all $a \in L$. Observe that $f$ depends on both $\mathcal{D}_1$
  and $\mathcal{D}_2$ and it is order preserving. Indeed, if $a \leq b$, then 
  $\mathcal{D}_1(r) \geq b$ implies $\mathcal{D}_1(r) \geq a$, and so
  \begin{align*}
    \{\mathcal{D}_2(r);\,
    r \in \mathrm{Tupl}(R) \text{ and } \mathcal{D}_1(r) \geq a\}
    \supseteq
    \{\mathcal{D}_2(r);\,
    r \in \mathrm{Tupl}(R) \text{ and } \mathcal{D}_1(r) \geq b\}
  \end{align*}
  from which it follows that
  \begin{align*}
    \inf\{\mathcal{D}_2(r);\,
    r \in \mathrm{Tupl}(R) \text{ and } \mathcal{D}_1(r) \geq a\}
    \leq
    \inf\{\mathcal{D}_2(r);\,
    r \in \mathrm{Tupl}(R) \text{ and } \mathcal{D}_1(r) \geq b\},
  \end{align*}
  i.e., $f(a) \leq f(b)$.
  Moreover, using~\eqref{eqn:circ} and~\eqref{eqn:fa}, we have
  \begin{align*}
    (\mathcal{D}_1 \circ f)(r) =
    f(\mathcal{D}_1(r)) =
    \textstyle\bigwedge\{\mathcal{D}_2(r');\,
    r' \in \mathrm{Tupl}(R)
    \text{ and }
    \mathcal{D}_1(r') \geq \mathcal{D}_1(r)\},
  \end{align*}
  i.e., in order to prove $\mathcal{D}_1 \circ f = \mathcal{D}_2$,
  it suffices to show that $\mathcal{D}_2(r)$ is the least element of
  \begin{align*}
    K = \{\mathcal{D}_2(r');\,
    r' \in \mathrm{Tupl}(R)
    \text{ and }
    \mathcal{D}_1(r') \geq \mathcal{D}_1(r)\}.
  \end{align*}
  Clearly, $\mathcal{D}_2(r) \in K$ owing to the reflexivity of $\geq$ in the
  special case of $r = r'$.
  Now, consider a general $\mathcal{D}_2(r') \in K$, i.e.,
  $r' \in \mathrm{Tupl}(R)$ such that $\mathcal{D}_1(r') \geq \mathcal{D}_1(r)$.
  Using~\eqref{eqn:U},
  it means $r' \in \mathcal{U}(\mathcal{D}_1,r)$ and so
  $r' \in \mathcal{U}(\mathcal{D}_2,r)$
  using the assumption $\mathcal{D}_1 \osub \mathcal{D}_2$.
  As a consequence, $\mathcal{D}_2(r') \geq \mathcal{D}_2(r)$,
  proving that $\mathcal{D}_2(r)$ is the least element of $K$
  which further gives $\mathcal{D}_1 \circ f = \mathcal{D}_2$.

  The if-part is easy to see: Let $f\!: L \to L$ be a map
  satisfying~\eqref{eqn:ord_pres} and $\mathcal{D}_1 \circ f = \mathcal{D}_2$.
  Take $r' \in \mathcal{U}(\mathcal{D}_1,r)$. Then,
  $\mathcal{D}_1(r') \geq \mathcal{D}_1(r)$ and so
  $f(\mathcal{D}_1(r')) \geq f(\mathcal{D}_1(r))$ because $f$
  is order preserving. Using $\mathcal{D}_1 \circ f = \mathcal{D}_2$,
  we obtain
  \begin{align*}
    \mathcal{D}_2(r') = f(\mathcal{D}_1(r')) \geq
    f(\mathcal{D}_1(r)) = \mathcal{D}_2(r),
  \end{align*}
  meaning $r' \in \mathcal{U}(\mathcal{D}_2,r)$.
  Hence, $\mathcal{U}(\mathcal{D}_1,r) \subseteq
  \mathcal{U}(\mathcal{D}_2,r)$ which proves
  $\mathcal{D}_1 \osub \mathcal{D}_2$.
\end{proof}

\begin{example}
  Consider RDTs $\mathcal{D}_1$ and $\mathcal{D}_2$ as in Example~\ref{ex:ord}.
  Since $\mathcal{D}_2 \osub \mathcal{D}_1$, Theorem~\ref{th:ord_sub} yields
  there is a map $f\!: L \to L$ such that $\mathcal{D}_1 = \mathcal{D}_2 \circ f$.
  A map $f$ satisfying this property is not given uniquely.
  The map given by~\eqref{eqn:fa} described in the proof of Theorem~\ref{th:ord_sub}
  is given by
  \begin{align*}
    f(a) &=
    \left\{
      \begin{array}{@{\,}l@{\quad}l@{}}
        0, &\text{if } a = 0, \\
        0.148, &\text{if } 0 < a \leq 0.148, \\
        0.426, &\text{if } 0.148 < a \leq 0.426, \\
        0.643, &\text{if } 0.426 < a \leq 0.643, \\
        0.778, &\text{if } 0.643 < a \leq 0.778, \\
        0.937, &\text{if } 0.778 < a \leq 0.939, \\
        1, &\text{if } a > 0.939, \\
      \end{array}
    \right.
  \end{align*}
  for all $a \in L$.
\end{example}

For the next theorem, recall that $L(\mathcal{D})$, called the range of $\mathcal{D}$,
represents the set of scores which appear in $\mathcal{D}$ and in general
it includes $0$, see~\eqref{eqn:LD}.

\begin{theorem}\label{th:ord_equ}%
  Let $\mathcal{D}_1$ and $\mathcal{D}_2$ be RDTs on $R$.
  Then, $\mathcal{D}_1 \oeqv \mathcal{D}_2$ if{}f there is
  an order isomorphism $f\!: L(\mathcal{D}_1) \to L(\mathcal{D}_2)$
  such that $\mathcal{D}_1 \circ f = \mathcal{D}_2$.
\end{theorem}
\begin{proof}
  Let $\mathcal{D}_1 \oeqv \mathcal{D}_2$, i.e., 
  $\mathcal{D}_1 \osub \mathcal{D}_2$ and $\mathcal{D}_2 \osub \mathcal{D}_1$.
  By Theorem~\ref{th:ord_sub}, there are order-preserving maps $f\!: L \to L$
  and $g\!: L \to L$ such that $\mathcal{D}_1 \circ f = \mathcal{D}_2$ and
  $\mathcal{D}_1 = \mathcal{D}_2 \circ g$. Furthermore, consider the
  restrictions $f|_{L(\mathcal{D}_1)}$ and $g|_{L(\mathcal{D}_2)}$ of $f$ and $g$
  to $L(\mathcal{D}_1)$ and $L(\mathcal{D}_2)$, respectively. Under this notation,
  $f|_{L(\mathcal{D}_1)}$ and $g|_{L(\mathcal{D}_2)}$ are order preserving maps
  of the form $f|_{L(\mathcal{D}_1)}\!: L(\mathcal{D}_1) \to L(\mathcal{D}_2)$ and 
  $g|_{L(\mathcal{D}_2)}\!: L(\mathcal{D}_2) \to L(\mathcal{D}_1)$
  which satisfy
  \begin{align*}
    \mathcal{D}_1 \circ f|_{L(\mathcal{D}_1)} = \mathcal{D}_2
    \text{ and }
    \mathcal{D}_1 = \mathcal{D}_2 \circ g|_{L(\mathcal{D}_2)}.
  \end{align*}
  Therefore, we have
  \begin{align*}
    \mathcal{D}_1(r) 
    &= \bigl(\mathcal{D}_2 \circ g|_{L(\mathcal{D}_2)}\bigr)(r)
    \\
    &= g|_{L(\mathcal{D}_2)}(\mathcal{D}_2(r))
    \\
    &= g|_{L(\mathcal{D}_2)}\bigl(
    \bigl(\mathcal{D}_1 \circ f|_{L(\mathcal{D}_1)}\bigr)(r)\bigr)
    \\
    &= g|_{L(\mathcal{D}_2)}\bigl(f|_{L(\mathcal{D}_1)}(\mathcal{D}_1(r))\bigr)
  \end{align*}
  for all $r \in \mathrm{Tupl}(R)$ which is in the answer set of $\mathcal{D}_1$.
  As a consequence, the composed map
  $f|_{L(\mathcal{D}_1)} \circ g|_{L(\mathcal{D}_2)}$ is the
  identity map on $L(\mathcal{D}_1)$. Using analogous arguments,
  $g|_{L(\mathcal{D}_2)} \circ f|_{L(\mathcal{D}_1)}$ is the identity
  map on $L(\mathcal{D}_2)$.
  This shows that $f|_{L(\mathcal{D}_1)}$ is an order embedding:
  $f|_{L(\mathcal{D}_1)}(a) \leq f|_{L(\mathcal{D}_1)}(b)$ implies
  $g|_{L(\mathcal{D}_2)}(f|_{L(\mathcal{D}_1)}(a)) \leq
  g|_{L(\mathcal{D}_2)}(f|_{L(\mathcal{D}_1)}(b))$ because
  $g|_{L(\mathcal{D}_2)}$ is order preserving and as a consequence
  of the fact that $f|_{L(\mathcal{D}_1)} \circ g|_{L(\mathcal{D}_2)}$ is the
  identity, we get that $a \leq b$.
  In addition, $f|_{L(\mathcal{D}_1)}\!: L(\mathcal{D}_1) \to L(\mathcal{D}_2)$
  is surjective because for each $a \in L(\mathcal{D}_2)$,
  we have that $f|_{L(\mathcal{D}_1)}(g|_{L(\mathcal{D}_2)}(a)) = a$.
  Altogether, $f|_{L(\mathcal{D}_1)}$ is the desired order isomorphism.

  In order to prove the if-part of Theorem~\ref{th:ord_equ}, let us
  consider an order isomorphism $f\!: L(\mathcal{D}_1) \to L(\mathcal{D}_2)$
  such that $\mathcal{D}_1 \circ f = \mathcal{D}_2$. Now, $f$ can be
  extended to a map $f^\sharp\!: L \to L$ by putting
  \begin{align*}
    f^\sharp(a) &= \textstyle\bigwedge\{f(a');\,
    a' \in L(\mathcal{D}_1) \text{ and } a' \geq a\}
  \end{align*}
  for all $a \in L$. Observe that $f^\sharp$ is indeed an extension of $f$:
  For $a \in L(\mathcal{D}_1)$, it follows that $f(a)$ belongs to
  \begin{align*}
    K &= \{f(a');\, a' \in L(\mathcal{D}_1) \text{ and } a' \geq a\}
  \end{align*}
  because of the reflexivity of $\geq$. Moreover, if $f(a') \in K$,
  then $a' \geq a$ and so $f(a') \geq f(a)$ owing to the fact that $f$
  is order preserving. Thus, $f(a)$ is the least element of $K$ and,
  as a consequence, $f^\sharp(a) = \textstyle\bigwedge\{f(a');\,
  a' \in L(\mathcal{D}_1) \text{ and } a' \geq a\} = f(a)$.
  Furthermore, the fact that $f$ is order preserving ensures
  that $f^\sharp$ is order preserving as well. Indeed,
  take any $a,b \in L$ such that $a \leq b$. Then,
  analogously as in the proof of Theorem~\ref{th:ord_sub},
  we have
  \begin{align*}
    \{f(a');\,
    a' \in L(\mathcal{D}_1) \text{ and } a' \geq a\}
    \supseteq
    \{f(a');\,
    a' \in L(\mathcal{D}_1) \text{ and } a' \geq b\}
  \end{align*}
  and thus 
  \begin{align*}
    \inf\{f(a');\,
    a' \in L(\mathcal{D}_1) \text{ and } a' \geq a\}
    \leq
    \inf\{f(a');\,
    a' \in L(\mathcal{D}_1) \text{ and } a' \geq b\}
  \end{align*}
  which proves $f^\sharp(a) \leq f^\sharp(b)$.
  Therefore, $\mathcal{D}_1 \osub \mathcal{D}_2$ owing
  to Theorem~\ref{th:ord_sub}. In addition,
  $\mathcal{D}_2 \osub \mathcal{D}_1$ follows using the same
  arguments using the inverse $f^{-1}$ of $f$. Note that $f$
  being an order isomorphism ensures that $f$ is a bijection,
  so the inverse of $f$ exists.
\end{proof}

\section{Invariance Theorems}\label{sec:inv}
In this section, we present two invariance theorems which are the main
observations of this paper. As a result of the invariance theorems,
it follows that results of arbitrary complex queries composed
of~\eqref{eqn:join}--\eqref{eqn:gsubs} are invariant to ordinal transformations:
If the input data are transformed by $f$ into ordinally equivalent data,
then the results of queries performed with the original and the new data
are also ordinally equivalent. As a practical consequence, if a transformation
of the input data does not change the order in which tuples appear in tables
when sorted by scores, then the same property holds for results of arbitrary
queries.

\begin{theorem}\label{th:inv1}%
  Let $f\!: L \to L$ be order preserving.
  Then, for any RDTs $\mathcal{D}_1,\mathcal{D}_2,\mathcal{D}$ for which
  both sides of the following equalities are defined, we have
  \begin{align}
    (\join{\mathcal{D}_1}{\mathcal{D}_2}) \circ f &= 
    \join{(\mathcal{D}_1 \circ f)}{(\mathcal{D}_2 \circ f)},
    \label{th:join}
    \\[2pt]
    \restrict{\rcond}{\mathcal{D}} \circ f &= 
    \restrict{\rcond \circ f}{\mathcal{D} \circ f},
    \label{th:restrict}
    \\[2pt]
    (\union{\mathcal{D}_1}{\mathcal{D}_2}) \circ f &= 
    \union{(\mathcal{D}_1 \circ f)}{(\mathcal{D}_2 \circ f)},
    \label{th:union}
    \\[2pt]
    \project{S}{\mathcal{D}} \circ f &= 
    \project{S}{\mathcal{D} \circ f}.
    \label{th:project}
  \end{align}
\end{theorem}
\begin{proof}
  In order to prove~\eqref{th:join}, we check that
  \begin{align*}
    f(\inf\{\mathcal{D}_1(rs),\mathcal{D}_2(st)\}) =
    \inf\{f(\mathcal{D}_1(rs)),f(\mathcal{D}_2(st))\}.
  \end{align*}
  Since $\mathbf{L}$ is linear, we may proceed by cases:
  First, assume that $\mathcal{D}_1(rs) \leq \mathcal{D}_2(st)$.
  Then, $f(\mathcal{D}_1(rs)) \leq f(\mathcal{D}_2(st))$ because $f$
  is order preserving and thus
  \begin{align*}
    f(\inf\{\mathcal{D}_1(rs),\mathcal{D}_2(st)\})
    &= f(\mathcal{D}_1(rs))
    \\
    &= \inf\{f(\mathcal{D}_1(rs))\}
    \\
    &= \inf\{f(\mathcal{D}_1(rs)),f(\mathcal{D}_2(st))\}.
  \end{align*}
  Second, assume $\mathcal{D}_1(rs) \geq \mathcal{D}_2(st)$
  and proceed as above with $\leq$ replaced by $\geq$.

  Analogously, we may proceed for \eqref{th:restrict}. It suffices to check
  that
  \begin{align*}
    f(\inf\{\mathcal{D}(r),\rcond(r)\}) =
    \inf\{f(\mathcal{D}(r)),\rcond(r)\}
  \end{align*}
  during which we distinguish two cases:
  (i)
  $\mathcal{D}(r) \leq \rcond(r)$ and thus $f(\mathcal{D}(r)) \leq f(\rcond(r))$;
  (ii)
  $\mathcal{D}(r) \geq \rcond(r)$ and $f(\mathcal{D}(r)) \geq f(\rcond(r))$.

  Now, \eqref{th:union} follows by the same argument
  as in the case of~\eqref{th:join} with $\sup$ in place of $\inf$.
  Indeed, we check that
  \begin{align*}
    f(\sup\{\mathcal{D}_1(r),\mathcal{D}_2(r)\}) =
    \sup\{f(\mathcal{D}_1(r)),f(\mathcal{D}_2(r))\}
  \end{align*}
  holds by cases in which we use the fact that 
  $\mathcal{D}_1(r) \leq \mathcal{D}_2(r)$ if{}f
  $\sup\{\mathcal{D}_1(r),\mathcal{D}_2(r)\} = \mathcal{D}_2(r)$
  together with the assumption that $f$ is order preserving,
  and dually for $\geq$.

  In case of~\eqref{th:project}, it suffices to
  check that $f$ commutes with suprema of finite subsets of $L$
  which is indeed the case. In a more detail,
  let $\mathcal{D}$ be an RDT on $R$ and $S \subseteq R$. In this
  setting, it suffices to prove that
  \begin{align*}
    f(\sup\{\mathcal{D}(st);\, t \in \mathrm{Tupl}(R \setminus S)\}) &=
    \sup\{f(\mathcal{D}(st));\, t \in \mathrm{Tupl}(R \setminus S)\}
  \end{align*}
  for any $s \in \mathrm{Tupl}(S)$. Observe that for any $s \in \mathrm{Tupl}(S)$,
  \begin{align*}
    K = \{\mathcal{D}(st);\, t \in \mathrm{Tupl}(R \setminus S)\}
  \end{align*}
  is a finite set of scores which is a subset of the (finite) range of $\mathcal{D}$.
  In addition, $K$ is non-empty because $\mathrm{Tupl}(R \setminus S)$ is always
  non-empty and $s$ and $t$ are trivially joinable. Therefore, owing to the fact
  that $\mathbf{L}$ is totally ordered, there is $t' \in \mathrm{Tupl}(R \setminus S)$
  such that $\mathcal{D}(st')$ is the greatest element of $K$. Since $f$ is order
  preserving, it readily follows that $f(\mathcal{D}(st'))$ is
  the greatest element of 
  \begin{align*}
    f(K) = \{f(\mathcal{D}(st));\, t \in \mathrm{Tupl}(R \setminus S)\}.
  \end{align*}
  Therefore, under this notation, we have
  \begin{align*}
    f(\sup\{\mathcal{D}(st);\, t \in \mathrm{Tupl}(R \setminus S)\})
    &=
    f(\sup K) \\
    &=
    f(\mathcal{D}(st')) \\
    &=
    \sup f(K) \\
    &=
    \sup\{f(\mathcal{D}(st));\, t \in \mathrm{Tupl}(R \setminus S)\},
  \end{align*}
  which proves~\eqref{th:project}.
\end{proof}

Under stronger assumptions than in Theorem~\ref{th:inv1},
we establish the following observation of invariance for
the remaining operations with RDTs.

\begin{theorem}\label{th:inv2}%
  Let $f\!: L \to L$ be order embedding and let
  $\mathcal{D}_1,\mathcal{D}_2,\mathcal{D}_3$ be RDTs for which
  both sides of the following equalities are defined. Then,
  \begin{align}
    (\rdiv{\mathcal{D}_3}{\mathcal{D}_1}{\mathcal{D}_2}) \circ f &=
    \rdiv{\mathcal{D}_3 \circ f}{(\mathcal{D}_1 \circ f)}{(\mathcal{D}_2 \circ f)},
    \label{th:divide}
    \\
    (\resd{\mathcal{D}_3}{\mathcal{D}_1}{\mathcal{D}_2}) \circ f &=
    \resd{\mathcal{D}_3 \circ f}{(\mathcal{D}_1 \circ f)}{(\mathcal{D}_2 \circ f)}.
    \label{th:resd}
    \intertext{If $f(0) = 0$, then}
    (\minus{\mathcal{D}_1}{\mathcal{D}_2}) \circ f &= 
    \minus{(\mathcal{D}_1 \circ f)}{(\mathcal{D}_2 \circ f)}.
    \label{th:minus}
    \intertext{If $f(1) = 1$, then}
    \gsubs{\mathcal{D}_1}{\mathcal{D}_2} \circ f &=
    \gsubs{\mathcal{D}_1 \circ f}{\mathcal{D}_2 \circ f}.
    \label{th:gsubs}
  \end{align}
\end{theorem}
\begin{proof}
  In case of~\eqref{th:minus},
  we distinguish two cases based on~\eqref{eqn:minus}.
  First, if we have $\mathcal{D}_1(r) \leq \mathcal{D}_2(r)$,
  then $f(\mathcal{D}_1(r)) \leq f(\mathcal{D}_2(r))$ because $f$
  is order preserving and so
  \begin{align*}
    ((\minus{\mathcal{D}_1}{\mathcal{D}_2}) \circ f)(r)
    &= f(0)
    = 0 
    = (\minus{(\mathcal{D}_1 \circ f)}{(\mathcal{D}_2 \circ f)})(r),
  \end{align*}
  taking into account the fact that $f(0) = 0$.
  Second, assume that $\mathcal{D}_1(r) \nleq \mathcal{D}_2(r)$.
  In this case, $\mathcal{D}_1(r) > \mathcal{D}_2(r)$ because $\mathbf{L}$
  is totally ordered and so $f(\mathcal{D}_1(r)) \geq f(\mathcal{D}_2(r))$ because
  $f$ is order preserving. Since $f$ is also order reflecting,
  we must have $f(\mathcal{D}_1(r)) > f(\mathcal{D}_2(r))$ because
  $f(\mathcal{D}_1(r)) = f(\mathcal{D}_2(r))$ would yield
  $\mathcal{D}_1(r) \leq \mathcal{D}_2(r)$, a contradiction.
  Therefore, we have $f(\mathcal{D}_1(r)) \nleq f(\mathcal{D}_2(r))$ and so
  \begin{align*}
    ((\minus{\mathcal{D}_1}{\mathcal{D}_2}) \circ f)(r)
    &=
    f((\minus{\mathcal{D}_1}{\mathcal{D}_2})(r))
    \\
    &=
    f(\mathcal{D}_1(r))
    \\
    &=
    (\minus{f(\mathcal{D}_1)}{f(\mathcal{D}_2)})(r),
  \end{align*}
  which proves~\eqref{th:minus}.

  In case of~\eqref{th:divide}, we may proceed by cases considering
  the condition~\eqref{eqn:div_cond}. In a more detail, let
  $\mathcal{D}_1$ be an RDT on $R \cup S$ such that $R \cap S = \emptyset$,
  $\mathcal{D}_2$ be an RDT on $S$, and $\mathcal{D}_3$ be an RDT on $R$.
  Furthermore, assume that for a given $r \in \mathrm{Tupl}(R)$ and
  all $s \in \mathrm{Tupl}(S)$, we have that
  $\mathcal{D}_2(s) > \mathcal{D}_1(rs)$ implies
  $\mathcal{D}_3(r) \leq \mathcal{D}_1(rs)$.
  In this case,
  $\bigl(\rdiv{\mathcal{D}_3}{\mathcal{D}_1}{\mathcal{D}_2}\bigr)(r) =
  \mathcal{D}_3(r)$.
  Moreover, the fact that $f$ is an order embedding gives that 
  $f(\mathcal{D}_2(s)) > f(\mathcal{D}_1(rs))$ implies 
  $\mathcal{D}_2(s) > \mathcal{D}_1(rs)$ and so
  $\mathcal{D}_3(r) \leq \mathcal{D}_1(rs)$, i.e.,
  $f(\mathcal{D}_3(r)) \leq f(\mathcal{D}_1(rs))$.
  As a consequence,
  \begin{align*}
    f\bigl(\bigl(\rdiv{\mathcal{D}_3}{\mathcal{D}_1}{\mathcal{D}_2}\bigr)(r)\bigr)
    &=
    f(\mathcal{D}_3(r))
    \\
    &=
    \bigl(\rdiv{\mathcal{D}_3 \circ f}{
      (\mathcal{D}_1 \circ f)}{(\mathcal{D}_2 \circ f)}\bigr).
  \end{align*}
  It remains to prove the equality in the case when~\eqref{eqn:div_cond}
  does not hold. That is, assume that for given $r \in \mathrm{Tupl}(R)$
  there is $s \in \mathrm{Tupl}(S)$ such that 
  $\mathcal{D}_2(s) > \mathcal{D}_1(rs)$ and
  $\mathcal{D}_3(r) > \mathcal{D}_1(rs)$.
  Therefore, for given $r \in \mathrm{Tupl}(R)$,
  \begin{align*}
    K &= 
    \{\mathcal{D}_1(rs);\,
    \mathcal{D}_2(s) > \mathcal{D}_1(rs)
    \text{, } s \in \mathrm{Tupl}(S)\}
  \end{align*}
  is non-empty and in addition it is finite because it is a subset of
  the range of $\mathcal{D}_1$. Since $\mathbf{L}$ is totally ordered,
  there is $s' \in \mathrm{Tupl}(S)$ such that $\mathcal{D}_1(rs')$
  is the least element of $K$. The fact that $f$ is an order embedding
  further gives that $f(\mathcal{D}_1(rs'))$ is the least element of
  \begin{align*}
    f(K) &= 
    \{f(\mathcal{D}_1(rs));\,
    f(\mathcal{D}_2(s)) > f(\mathcal{D}_1(rs))
    \text{, } s \in \mathrm{Tupl}(S)\}.
  \end{align*}
  Hence,
  \begin{align*}
    f\bigl(\bigl(\rdiv{\mathcal{D}_3}{\mathcal{D}_1}{\mathcal{D}_2}\bigr)(r)\bigr)
    &=
    f(\mathcal{D}_1(rs'))
    \\
    &=
    \bigl(\rdiv{\mathcal{D}_3 \circ f}{
      (\mathcal{D}_1 \circ f)}{(\mathcal{D}_2 \circ f)}\bigr),
  \end{align*}
  which concludes the proof of~\eqref{th:divide}.
  Now, observe that~\eqref{th:gsubs} follows directly by~\eqref{th:divide}.
  Indeed, for $\mathcal{D}_1$ and $\mathcal{D}_2$ on $R$
  and for an auxiliary $\mathcal{D}$ on $\emptyset$ such that
  $\mathcal{D}(\emptyset) = 1$, we have
  \begin{align*}
    \gsubs{\mathcal{D}_1}{\mathcal{D}_2} \circ f &=
    f\bigl(\bigl(\rdiv{\mathcal{D}}{
      \mathcal{D}_2}{\mathcal{D}_1}\bigr)(\emptyset)\bigr)
    \\
    &=
    \bigl(\rdiv{\mathcal{D} \circ f}{
      (\mathcal{D}_2 \circ f)}{(\mathcal{D}_1 \circ f)}\bigr)(\emptyset)
    \\
    &=
    \bigl(\rdiv{\mathcal{D}}{
      (\mathcal{D}_2 \circ f)}{(\mathcal{D}_1 \circ f)}\bigr)(\emptyset)
    \\
    &=
    \gsubs{\mathcal{D}_1 \circ f}{\mathcal{D}_2 \circ f}
  \end{align*}
  provided that $f(1) = 1$ and thus $\mathcal{D} \circ f = \mathcal{D}$.
  Finally, \eqref{th:resd} can be proved analogously as~\eqref{th:divide} by
  inspecting the cases in~\eqref{eqn:resd_cases},
  the details are left to the reader.
\end{proof}

If $f\!: L \to L$ is an order isomorphism, then all
conditions in Theorem~\ref{th:inv1} and Theorem~\ref{th:inv2}
are satisfied including the facts that $f(0) = 0$
and $f(1) = 1$. Such $f$ may be viewed as an ordinal
transformation function of ranked data tables.
We may say that $\mathcal{D}_1$ is ordinally transformed into
$\mathcal{D}_2$ by $f$,
written $\mathcal{D}_1 \mapsto_f \mathcal{D}_2$, whenever
$\mathcal{D}_1 \circ f = \mathcal{D}_2$. Under this notation,
\eqref{th:join}--\eqref{th:gsubs} in the invariance theorems
can be restated as follows: If
$\mathcal{D}_1 \mapsto_f \mathcal{D}'_1$ and
$\mathcal{D}_2 \mapsto_f \mathcal{D}'_2$, then
\begin{align}
  \join{\mathcal{D}_1}{\mathcal{D}_2} \oeqv
  \join{\mathcal{D}'_1}{\mathcal{D}'_2}
\end{align}
in case of $\bowtie$ and analogously for
$\sigma_{\!\rcond}$, $\pi_{\!S}$, $\cup$, $-$, $\frdiv$, and $\gsubsop$.
Put in words, the results of an operation with transformed input
data and the original input data are equivalent in terms of
the order of tuples given by scores.

\begin{figure}
  \centering
  \begin{relation}{2}
    \rank &
    \atr{ID} &
    \atr{PRICE} \\
    \hline
    \rule{0pt}{9pt}%
    $0.882$ & \val{71} & \val{\usd{798,000}} \\
    $0.882$ & \val{71} & \val{\usd{849,000}} \\
    $0.655$ & \val{85} & \val{\usd{998,000}} \\
    $0.541$ & \val{82} & \val{\usd{648,000}} \\
    $0.462$ & \val{58} & \val{\usd{829,000}} \\
    $0.272$ & \val{93} & \val{\usd{598,000}} \\
  \end{relation}
  \quad
  \begin{relation}{2}
    \rank &
    \atr{ID} &
    \atr{PRICE} \\
    \hline
    \rule{0pt}{9pt}%
    $0.877$ & \val{71} & \val{\usd{798,000}} \\
    $0.782$ & \val{71} & \val{\usd{849,000}} \\
    $0.655$ & \val{85} & \val{\usd{998,000}} \\
    $0.429$ & \val{58} & \val{\usd{829,000}} \\
    $0.361$ & \val{82} & \val{\usd{648,000}} \\
    $0.160$ & \val{93} & \val{\usd{598,000}} \\
  \end{relation}
  \caption{Join of transformed ranked data table projected onto
    $\{\atr{id},\atr{price}\}$ (left) and result of the same query using
    the Goguen aggregation (right).}
  \label{eqn:f_GoGo}
\end{figure}

\begin{figure}
  \centering
  \begin{relation}{3}
    \rank &
    \atr{ID} &
    \atr{BDRM} &
    \atr{PRICE} \\
    \hline
    \rule{0pt}{9pt}%
    $0.778$ & \val{85} & \val{5} & \val{\usd{998,000}} \\
    $0.699$ & \val{71} & \val{3} & \val{\usd{798,000}} \\
    $0.699$ & \val{71} & \val{3} & \val{\usd{849,000}} \\
    $0.643$ & \val{82} & \val{4} & \val{\usd{648,000}} \\
    $0.426$ & \val{58} & \val{4} & \val{\usd{829,000}} \\
    $0.148$ & \val{93} & \val{2} & \val{\usd{598,000}} \\
  \end{relation}
  \quad
  \begin{relation}{3}
    \rank &
    \atr{ID} &
    \atr{BDRM} &
    \atr{PRICE} \\
    \hline
    \rule{0pt}{9pt}%
    $0.655$ & \val{85} & \val{5} & \val{\usd{998,000}} \\
    $0.579$ & \val{71} & \val{3} & \val{\usd{798,000}} \\
    $0.579$ & \val{71} & \val{3} & \val{\usd{849,000}} \\
    $0.541$ & \val{82} & \val{4} & \val{\usd{648,000}} \\
    $0.462$ & \val{58} & \val{4} & \val{\usd{829,000}} \\
    $0.272$ & \val{93} & \val{2} & \val{\usd{598,000}} \\
  \end{relation}
  \caption{Result of $\restrict{\rcond}{\join{\mathcal{D}_1}{\mathcal{D}_2}}$
    projected onto $S = \{\atr{id},\atr{bdrm},\atr{price}\}$ (left) and
    $\restrict{\rcond \circ f}{\join{(\mathcal{D}_1 \circ f)}{(\mathcal{D}_2} \circ f)}$
    projected onto $S$ (right).}
  \label{eqn:f_select}
\end{figure}

\begin{example}
  To illustrate the invariance theorems on concrete data,
  consider the RDTs $\mathcal{D}_1$ and $\mathcal{D}_2$ as in Fig.\,\ref{fig:RDTs}.
  A map $f\!: [0,1] \to [0,1]$ given by
  \begin{align}
    f(x) &=
    \left\{
      \begin{array}{@{\,}l@{\quad}l@{}}
        2^{-0.5}\sqrt{x}, &\text{if } x \leq 0.5, \\[2pt]
        2(x-0.5)^2 + 0.5, &\text{otherwise,}
      \end{array}
    \right.
  \end{align}
  is an order isomorphism preserving $0$ and $1$.
  Fig.\,\ref{eqn:f_GoGo} (left) contains the result of
  \begin{align*}
    \project{\{\atr{id},\atr{price}\}}{%
      \join{(\mathcal{D}_1 \circ f)}{(\mathcal{D}_2 \circ f)}}
  \end{align*}
  which is equivalent to
  \begin{align*}
    \project{\{\atr{id},\atr{price}\}}{%
      \join{\mathcal{D}_1}{\mathcal{D}_2}} \circ f
  \end{align*}
  owing to \eqref{th:join} and \eqref{th:restrict}.
  The tuples in the result, when sorted by scores,
  appear in the same order as in Fig.\,\ref{fig:join}
  showing $\join{\mathcal{D}_1}{\mathcal{D}_2}$. Our assumption
  that the join~\eqref{eqn:join} (as well as the other operations) is
  defined using the infimum instead of a general aggregation
  function $\otimes$, see Remark~\ref{rem:general}, is essential.
  If we replace $\inf$ in~\eqref{eqn:join}
  by $\otimes$ being the multiplication of reals
  (so-called Goguen aggregation, see \cite{Gog:Lic})
  and compute 
  $\project{\{\atr{id},\atr{price}\}}{%
    \join{(\mathcal{D}_1 \circ f)}{(\mathcal{D}_2 \circ f)}}$,
  we get Fig.\,\ref{eqn:f_GoGo} (right)
  as the result where the order of tuples is not preserved.

  As a further example, Fig.\,\ref{eqn:f_select}\,(left) shows the
  result of a restriction of the join using the restriction
  condition $\rcond$ defined by
  \begin{align}
    \rcond(r) &=
    \left\{
      \begin{array}{@{\,}l@{\quad}l@{}}
        0.1 (4 + r(\atr{BDRM})), &\text{if } r(\atr{DBRM}) \leq 6, \\
        1, &\text{otherwise,} 
      \end{array}
    \right.
  \end{align}
  which may be seen as a restriction on the number of bedrooms $6$ and
  more with a tolerance for lower numbers. Fig.\,\ref{eqn:f_select} (right)
  shows the result for the tables transformed by $f$ as above. Again, the
  tuples appear in the same order. Finally, Fig.\,\ref{eqn:f_novarphi}
  shows that without transforming $\rcond$, the order of tuples in the
  result would not be preserved, i.e., $\rcond \circ f$
  in~\eqref{th:restrict} cannot be replaced by $\rcond$.
\end{example}

\begin{figure}
  \centering
  \begin{relation}{3}
    \rank &
    \atr{ID} &
    \atr{BDRM} &
    \atr{PRICE} \\
    \hline
    \rule{0pt}{9pt}%
    $0.699$ & \val{71} & \val{3} & \val{\usd{798,000}} \\
    $0.699$ & \val{71} & \val{3} & \val{\usd{849,000}} \\
    $0.655$ & \val{85} & \val{5} & \val{\usd{998,000}} \\
    $0.541$ & \val{82} & \val{4} & \val{\usd{648,000}} \\
    $0.462$ & \val{58} & \val{4} & \val{\usd{829,000}} \\
    $0.272$ & \val{93} & \val{2} & \val{\usd{598,000}} \\
  \end{relation}
  \caption{Result of $\project{S}{\restrict{\rcond}
      {\join{(\mathcal{D}_1 \circ f)}{(\mathcal{D}_2} \circ f)}}$.}
  \label{eqn:f_novarphi}
\end{figure}

The invariance theorems can be seen as type of description of the independence
of query results on possible changes in scores in the input data and restriction
conditions in queries. An alternative characterization which does not utilize
the position of tuples in relations but uses a notion of similarity was proposed
in~\cite{BeUrVy:Sadrqlor}. We now make a comment on how the approaches can be
combined. As we have outlined in the introduction,
\cite{BeUrVy:Sadrqlor} introduces lower bounds for similarity of query results
based on similarity of input data. For instance, in the case of joins
of RDTs, \cite{BeUrVy:Sadrqlor} shows that
\begin{align}
  \gsubs{\mathcal{D}_1}{\mathcal{D}_2} \otimes
  \gsubs{\mathcal{D}_3}{\mathcal{D}_4} \leq
  \gsubs{\join{\mathcal{D}_1}{\mathcal{D}_3}}{\join{\mathcal{D}_2}{\mathcal{D}_4}},
  \label{eqn:sim_est1}
  \\
  \gequs{\mathcal{D}_1}{\mathcal{D}_2} \otimes
  \gequs{\mathcal{D}_3}{\mathcal{D}_4} \leq
  \gequs{\join{\mathcal{D}_1}{\mathcal{D}_3}}{\join{\mathcal{D}_2}{\mathcal{D}_4}},
  \label{eqn:eqv_est1}
\end{align}
where $\gsubsop$ is defined as in \eqref{eqn:gsubs_general} and
$\gequsop$ is defined as in~\eqref{eqn:gequs}, and $\otimes$ is a binary
aggregation function with suitable properties (it is commutative, associative,
and $1$ is its neutral element). In our setting, \eqref{eqn:sim_est1}
and~\eqref{eqn:eqv_est1} may be restated with $\otimes$ replaced by $\inf$ as
\begin{align}
  \inf\{\gsubs{\mathcal{D}_1}{\mathcal{D}_2},
  \gsubs{\mathcal{D}_3}{\mathcal{D}_4}\} \leq
  \gsubs{\join{\mathcal{D}_1}{\mathcal{D}_3}}{\join{\mathcal{D}_2}{\mathcal{D}_4}},
  \label{eqn:Slb}
  \\
  \inf\{\gequs{\mathcal{D}_1}{\mathcal{D}_2},
  \gequs{\mathcal{D}_3}{\mathcal{D}_4}\} \leq
  \gequs{\join{\mathcal{D}_1}{\mathcal{D}_3}}{\join{\mathcal{D}_2}{\mathcal{D}_4}}.
  \label{eqn:Elb}
\end{align}
Put in words, \eqref{eqn:Slb} says that the score to which
$\join{\mathcal{D}_1}{\mathcal{D}_3}$ is contained in
$\join{\mathcal{D}_2}{\mathcal{D}_4}$ as at least the score to
which $\mathcal{D}_1$ is contained in $\mathcal{D}_2$ and 
$\mathcal{D}_3$ is contained in $\mathcal{D}_4$. Analogously,
we may interpret~\eqref{eqn:Elb} with ``contained'' replaced by ``similar''.

Now, using the fact that $f(\inf\{a,b\}) = \inf\{f(a),f(b)\}$ for all $a,b \in L$
together with the fact that $f$ is order-preserving, we may conclude that
\begin{align*}
  \inf\{\gsubs{\mathcal{D}_1 \circ f}{\mathcal{D}_2 \circ f},
  \gsubs{\mathcal{D}_3 \circ f}{\mathcal{D}_4 \circ f}\} &=
  \inf\{\gsubs{\mathcal{D}_1}{\mathcal{D}_2} \circ f,
  \gsubs{\mathcal{D}_3}{\mathcal{D}_4} \circ f\}
  \\
  &=
  f(\inf\{\gsubs{\mathcal{D}_1}{\mathcal{D}_2},
  \gsubs{\mathcal{D}_3}{\mathcal{D}_4}\})
  \\
  &\leq
  f(\gsubs{\join{\mathcal{D}_1}{\mathcal{D}_3}}{\join{\mathcal{D}_2}{\mathcal{D}_4}})
  \\
  &=
  \gsubs{\join{\mathcal{D}_1}{\mathcal{D}_3}}{\join{\mathcal{D}_2}{\mathcal{D}_4}}
  \circ f
\end{align*}
and analogously for $\gequsop$. In much the same way,
we get the following inequality:
\begin{align*}
  f(\inf\{\gsubs{\mathcal{D}_1}{\mathcal{D}_2},
  \gsubs{\mathcal{D}_3}{\mathcal{D}_4}\}) &=
  \inf\{\gsubs{\mathcal{D}_1}{\mathcal{D}_2} \circ f,
  \gsubs{\mathcal{D}_3}{\mathcal{D}_4} \circ f\}
  \\
  &=
  \inf\{\gsubs{\mathcal{D}_1 \circ f}{\mathcal{D}_2 \circ f},
  \gsubs{\mathcal{D}_3 \circ f}{\mathcal{D}_4 \circ f}\}
  \\
  &\leq
  \gsubs{\join{(\mathcal{D}_1 \circ f)}{(\mathcal{D}_3 \circ f)}}{
    \join{(\mathcal{D}_2 \circ f)}{(\mathcal{D}_4 \circ f)}}.
\end{align*}
The inequality may be seen as an extension of the lower bound given
by~\eqref{eqn:Slb} which incorporates an ordinal transformation.
Indeed, it reads: ``the score to which the join of the transformed RDTs
$\mathcal{D}_1$ and $\mathcal{D}_3$ is contained in the join of the 
transformed RDTs $\mathcal{D}_2$ and $\mathcal{D}_4$ is at least as
high as the transformed score of containment of 
$\mathcal{D}_1$ in $\mathcal{D}_2$ and
$\mathcal{D}_3$ in $\mathcal{D}_4$.'' Analogous combined similarity bounds
of operation with transformed data can be obtained for the other relational
operations, cf.~\cite{BeUrVy:Sadrqlor}.

\section{G\"odel logic and relational calculi}\label{sec:heyting}
In the previous section, we have discussed the invariance to ordinal transformations
for one particular query system---a system based on relational operations which may
be composed to form complex queries. The system resembles the traditional relational
algebra. In this section, we show that the same type of results on invariance to
ordinal transformations can also be established in a query system which is based
on evaluating formulas in database instances consisting of ranked data tables
and is conceptually similar to the classic relational calculi. We establish the
invariance theorems indirectly by showing that the query system based on evaluating
formulas is equivalent to the system based on relational operations. By proving the
equivalences of the query systems, we get new insights into the original query system.
For instance, it turns our that G\"odel logic plays an analogous role in the
rank-aware approach investigated in this paper as the Boolean logic in the classic
relational model of data. This connection allows us to derive conclusions about
properties of the relational operations in our model based on provability of
particular formulas in G\"odel logic---we utilize this observation
in Section~\ref{sec:comp}. In the beginning of this section, we recall
first-order G\"odel logic in a form that is suitable for our development
and then we show its relationship to our model.

A language $\mathcal{J}$ of a first-order G\"odel logic is given by a set
of relation symbols together with information about their arities. The relation
symbols may also be called predicate symbols and in the database terminology they
may be understood as relation variables whose values are bound to relations in 
database instances. Furthermore, we consider a denumerable set $X$ of object variables.
Analogously as in the case of the classic first-order logic, formulas are defined
recursively based on atomic formulas using symbols for logical connectives and
quantifiers:
\begin{enumlist}\parskip=\smallskipamount%
\item[\itm{1}]
  $\False$ is a formula (a constant of the truth value ``false'').
\item[\itm{2}]
  If $\rsym$ is $n$-ary relation symbol and $x_1,\ldots,x_n \in X$,
  then $\rsym(x_1,\ldots,x_n)$ is a formula.
\item[\itm{3}]
  If $\varphi$ and $\psi$ are formulas,
  then $(\varphi \logand \psi)$ and $(\varphi \logimp \psi)$ are formulas.
\item[\itm{4}]
  If $\varphi$ is a formula and $x \in X$,
  then $(\forall x)\varphi$ and $(\exists x)\varphi$ are formulas.
\end{enumlist}
All formulas we consider result by applications of \itm{1}--\itm{4}. Let us note
that both $\itm{1}$ and $\itm{2}$ introduce atomic formulas. In the first case,
$\False$ may be seen as a nullary logical connective (i.e., a connective
with no arguments). In the second case, each $\rsym(x_1,\ldots,x_n)$ is an atomic
formula constructed as in the first-order Boolean logic except for the fact
that we do not consider more complex terms than object variables---objects constants
and general function symbols may also be introduced but this is not necessary
for our application of the logic. Also note that a special case of \itm{2}
are formulas of the form $\rsym()$ when $\rsym$ is a nullary relation symbol.
In such a case, $\rsym()$ may be denoted just $\rsym$ and called a propositional symbol.
Furthermore, \itm{3} introduces more complex formulas built using logical connectives
$\logand$ (conjunction) and $\logimp$ (implication);
here we adopt the common rules for omission of outer parentheses in formulas.
Finally, \itm{4} defines universally and existentially quantified formulas in
the same way as in the classic logic.

\begin{remark}
  We can consider only $\False$, $\logand$, and $\logimp$ as the basic
  connectives. Indeed, formulas containing $\logor$ (disjunction) and possibly
  other connectives ($\logneg$ for a negation, and $\logequ$ for a biconditional)
  can be seen as abbreviations as follows:
  \begin{align}
    \logneg \varphi
    &\text{ is }
    \varphi \logimp \False,
    \\
    \varphi \logequ \psi
    &\text{ is }
    (\varphi \logimp \psi) \logand (\psi \logimp \varphi),
    \\
    \varphi \logor \psi
    &\text{ is }
    ((\varphi \logimp \psi) \logimp \psi) \logand
    ((\psi \logimp \varphi) \logimp \varphi).
    \label{eqn:short_or}
  \end{align}
  Note that in G\"odel logic, $\logand$ is not definable based solely on $\logimp$
  and $\False$ as it is in the classical logic where $\varphi \logand \psi$ can be
  seen as an abbreviation for $(\varphi \logimp (\psi \logimp \False)) \logimp \False$.
  This is due to the absence of the law of the double negation.
\end{remark}

The semantic of formulas is introduced based on their evaluation in general
structures for a given language $\mathcal{J}$ based on G\"odel algebras. In the
database terminology, the language defines a database scheme and the general
structures may be seen as counterparts to the classic database instances.

Let $\mathbf{L}$ be a G\"odel algebra. An $\mathbf{L}$-structure for
language $\mathcal{J}$ is denoted $\mathbf{M}$ and consists of a non-empty universe
set $M$ and a set which contains, for each $n$-ary relation symbol $\rsym$ in
the language,
a map $\rsym^\mathbf{M}\!: M^n \to L$ where $M^n$ denotes the usual $n$-th power of $M$.
Under this notation, $\rsym^\mathbf{M}(m_1,\ldots,m_n)$ is a degree in $L$ which can be
interpreted as a score of a tuple consisting of the values $m_1,\ldots,m_n$ in
$\rsym^\mathbf{M}$. Note that in this setting, we do not have names of attributes and
therefore the order of arguments in $\rsym^\mathbf{M}(m_1,\ldots,m_n)$ matters (as it
is usual in first-order logics, one may easily introduce ``names of attributes''
to keep the formalism closer to the style of relational database calculi).
An $\mathbf{M}$-valuation (of object variables) is any map $v\!: X \to M$,
$v(x)$ interpreted as the value of $x \in X$ under $v$. Now, the values of
formulas (of the language $\mathcal{J}$) in $\mathbf{L}$-structure $\mathbf{M}$
(for $\mathcal{J}$) given an $\mathbf{M}$-valuation $v$ is defined by the
following rules. In case of the atomic formulas, we put
\begin{align}
  ||\False||_{\mathbf{M},v} &= 0, 
  \\
  ||\rsym(x_1,\ldots,x_n)||_{\mathbf{M},v} &=
  \rsym^\mathbf{M}(v(x_1),\ldots,v(x_n)).
\end{align}
For the formulas built using the binary connectives $\logand$ and $\logimp$,
we put
\begin{align}
  ||\varphi \logand \psi||_{\mathbf{M},v} &=
  \inf \{||\varphi||_{\mathbf{M},v}, ||\psi||_{\mathbf{M},v}\},
  \\
  ||\varphi \logimp \psi||_{\mathbf{M},v} &=
  ||\varphi||_{\mathbf{M},v} \rightarrow ||\psi||_{\mathbf{M},v},
\end{align}
From~\eqref{eqn:short_or} it follows that
\begin{align}
  ||\varphi \logor \psi||_{\mathbf{M},v} &=
  \sup\{||\varphi||_{\mathbf{M},v}, ||\psi||_{\mathbf{M},v}\}.
\end{align}
Observe that if $\mathbf{L}$ is totally ordered, then~\eqref{eqn:resd_chain} yields
\begin{align}
  ||\varphi \logimp \psi||_{\mathbf{M},v} &=
  \begin{cases}
    1, & \text{if } ||\varphi||_{\mathbf{M},v} \leq ||\psi||_{\mathbf{M},v}, \\
    ||\psi||_{\mathbf{M},v}, &\text{otherwise.}
  \end{cases}
\end{align}
Finally, the value of quantified formulas is defined as follows provided that
the right-hand sides of the following equalities are defined:
\begin{align}
  ||(\forall x)\varphi||_{\mathbf{M},v} &=
  \inf\{||\varphi||_{\mathbf{M},w};\, w =_x v\},
  \label{eqn:forall}
  \\
  ||(\exists x)\varphi||_{\mathbf{M},v} &=
  \sup\{||\varphi||_{\mathbf{M},w};\, w =_x v\},
  \label{eqn:exists}
\end{align}
where $w =_x v$ means that $w$ is an $\mathbf{M}$-valuation such that $w(y) = v(y)$
for all $y \in X$ such that $x \ne y$. Note that in general,
\eqref{eqn:forall} and~\eqref{eqn:exists} may not be defined
because of the non-existence of infima and suprema of
$\{||\varphi||_{\mathbf{M},w};\, w =_x v\} \subseteq L$.
If for any $\varphi$ of the language $\mathcal{J}$ \eqref{eqn:forall}
and~\eqref{eqn:exists} are defined under any $\mathbf{M}$-valuation,
then $\mathbf{M}$ is called safe. If $\mathbf{L}$ is complete,
then any $\mathbf{L}$-structure is trivially safe. More importantly,
if each $r^{\mathbf{M}}$ is finite, meaning there are only finitely
many $m_1,\ldots,m_n \in M$ for which $r^{\mathbf{M}}(m_1,\ldots,m_n) > 0$,
then $\mathbf{M}$ is safe as well.

At this point, we can already describe how the interpretation of formulas in
G\"odel logic can be used as a basis of a query system and put it with
correspondence to the query system based on relational operations. We describe
the query system only to the extent to be able to derive conclusions on 
the invariance to ordinal transformations because a detailed description
of relational calcluli is beyond the scope and need of this paper. Interested
readers can find more details on pseudo-tuple calculus in~\cite{VaVy:Rdrad}.

Now, consider any finite $\mathbf{L}$-structure $\mathbf{M}$
(i.e., every $\rsym^\mathbf{M}$
is finite in the same sense as above) and a formula $\varphi$ with free variables
$x_1,\ldots,x_n$. Under this notation, $\mathbf{M}$ and $\varphi$ induce
a map $\mathcal{D}_{\mathbf{M},\varphi}\!: \mathrm{Tupl}(R) \to L$,
where $R = \{x_1,\ldots,x_n\}$ and
\begin{align}
  \bigl(\mathcal{D}_{\mathbf{M},\varphi}\bigr)(r) &= ||\varphi||_{\mathbf{M},v}
  \label{eqn:Mvarphi_to_D}
\end{align}
such that $r(x_i) = v(x_i)$ for all $i=1,\ldots,n$.
Clearly, $\mathcal{D}_{\mathbf{M},\varphi}$ given by~\eqref{eqn:Mvarphi_to_D}
is a ranked data table on $R$ (free variables in $\varphi$ are considered as
names of attributes) and it can be seen as a result of a query given by $\varphi$
in a database instance represented by the safe $\mathbf{L}$-structure $\mathbf{M}$.

\begin{remark}
  Let us note that $\mathcal{D}_{\mathbf{M},\varphi}$ is defined correctly
  by~\eqref{eqn:Mvarphi_to_D} because $||\varphi||_{\mathbf{M},v}$ depends only
  on $\mathbf{M}$-valuation of variables which appear free in $\varphi$. Also note
  that in the definition of $\mathcal{D}_{\mathbf{M},\varphi}$, we have tacitly
  assumed that variables in $\varphi$ are used as attribute names and, at the same
  time, we have disregarded their types. An explicit (and rigorous) treatment of
  types can be incorporated but it does not bring new insight into the invariance
  issues and we therefore use this simplification. The role of $\mathbf{L}$-structures
  as database instances is basically the same as in the classic model except for
  the fact that each $\rsym^\mathbf{M}$ represents an RDT instead of
  a classic relation. Indeed, a propositional symbol $\rsym$ may be seen as a name
  and $\rsym^\mathbf{M}$ (the interpretation of $\rsym$ in $\mathbf{M}$) may be
  seen as a current value of $\rsym$ considering $\mathbf{M}$.
\end{remark}

The equality of the considered query systems can be proved by showing that for
a query formulated in one of the systems there is a corresponding equivalent query
in the second one and \emph{vice versa.} The arguments are similar as in the ordinary
non-ranked model and we therefore focus only on the essential differences. 

\paragraph{From Relational Operations to Queries in G\"odel Logic}
Let us assume that $\mathcal{D}$ is a result of a query which uses
RDTs $\mathcal{D}_1,\ldots,\mathcal{D}_n$, restriction conditions
$\rcond_1,\ldots,\rcond_k$, and operations $\bowtie$, $\sigma$,
$\pi$, $\cup$, $\frdiv$, and renaming (in the ordinary sense).
Then, there is a finite $\mathbf{M}$ and a formula $\varphi$
such that $\mathcal{D}$ coincides with $\mathcal{D}_{\mathbf{M},\varphi}$.
The construction of $\mathbf{M}$ and $\varphi$ is straightforward and
goes along the same lines as in the ordinary case except for the fact
that the division is not a derivable operation. First, let $\mathbf{M}$
be an RDT where each RDT $\mathcal{D}_i$ is represented by $\rsym_i^\mathbf{M}$
and each restriction condition $\rcond_j$ is represented by $\rsym[s]_j^\mathbf{M}$.
Observe that since we consider only finitely many input RDTs, the universe of
$\mathbf{M}$ can be considered as a finite set and all $\rsym[s]_j^\mathbf{M}$'s
can be restricted to this finite universe. In case of queries resulting by
$\bowtie$, $\sigma$, $\pi$, and $\cup$, the desired formula is constructed
as in the classic case from formulas corresponding to subqueries.
For instance, let us assume that the query is of the form of
a projection onto $S = \{y_1,\ldots,y_p\}$ for $p \geq 0$ and its
subquery (the argument for the projection) produces an RDT
on $R = \{y_1,\ldots,y_q\}$ for $q \geq p$.
If we assume that a formula $\psi$ is a counterpart to the subquery, then
the counterpart of the projection is
\begin{align}
  (\exists y_{p+1})\cdots(\exists y_q)\psi,
  \label{eqn:cntr_proj}
\end{align}
i.e., the same formula as in the classic case. In case of the division,
which is not a fundamental operation, we proceed analogously. Namely,
we use a formula
\begin{align}
  \vartheta \logand (\forall y_1)\cdots(\forall y_p)(\psi \logimp \chi),
  \label{eqn:cntr_div}
\end{align}
where $\psi$, $\chi$, and $\vartheta$ are formulas corresponding to subqueries,
and $\{y_1,\ldots,y_p\}$ is the set of all attributes which are common
to the results of subqueries corresponding to $\psi$ and $\chi$,
cf.~\eqref{eqn:div}. Altogether, query of arbitrary complexity formulated
in terms of the relational operations with RDTs can equivalently be expressed
by a formula of G\"odel logic.

\paragraph{From Queries in G\"odel Logic to Relational Operations}
Conversely, consider any
finite $\mathbf{L}$-structure $\mathbf{M}$ with a universe $M$
and a formula $\varphi$.
Let $\mathcal{D}_M$ denote an RDT on $\{y\}$ such that
$\{r(y);\, \mathcal{D}_M(r) = 1\} = M$ and $L(\mathcal{D}_M) = \{1\}$.
Since $\mathbf{M}$ is finite, such an RDT always exists. Let $0_R$ denote
an empty RDT on $R$ (i.e., the answer set of $0_R$ is empty). Under this notation,
one can construct a relational expression which involves $\mathcal{D}_M$,
finitely many RDTs $0_R$, RDTs corresponding to all $\rsym^\mathbf{M}$, and
relational operations $\bowtie$, $\sigma$, $\pi$, $\cup$, $\rightarrow$,
and $\frdiv$ such that $\mathcal{D}_{\mathbf{M},\varphi}$ coincides with
the value of the expression. 
Again, the construction is fully analogous to the classic one except for the
fact that we consider two fundamental
quantifiers and fundamental connectives $\False$ and $\logimp$ which cannot be
defined in terms of the other ones. Indeed, if $\varphi$ is $\False$,
the corresponding expression is $0_\emptyset$, i.e., $0_R$ for $R = \emptyset$.
If $\varphi$ is $\rsym(x_1,\ldots,x_n)$, then we can consider the RDT corresponding
to $\rsym^\mathbf{M}$. For $\varphi$ being either of $\psi \logand \chi$ and
$\psi \logimp \chi$, we utilize the relational operations $\bowtie$
and $\rightarrow$ in conjunction with $\mathcal{D}_M$ (and optionally
the renaming of attributes); note here that as in the classic case,
$\psi$ and $\chi$ may have different sets
of variables which appear free in $\psi$ and $\chi$, respectively.
If $\varphi$ is $(\exists x)\psi$, we proceed as in the classic
case using $\pi$ and $\bowtie$. For $\varphi$ being $(\forall x)\psi$,
the expression is built using $\frdiv$ and $\mathcal{D}_M$. Namely,
the expression is of the form $\rdiv{\mathcal{D}}{\mathcal{D}_\psi}{\mathcal{D}_x}$,
where
(i)
$\mathcal{D}_x$ is $\mathcal{D}_M$ with the attribute $y$ renamed to $x$;
(ii)
$\mathcal{D}$ is a join of finitely many $\mathcal{D}_{y_1},\ldots,\mathcal{D}_{y_n}$
such that all free variables in $\psi$ except for $x$ are exactly $y_1,\ldots,y_n$
and each $\mathcal{D}_{y_i}$ results from $\mathcal{D}_M$ by renaming $y$ to $y_i$;
(iii)
$\mathcal{D}_\psi$ results by the expression corresponding to $\psi$.

\medskip
Owing to the correspondence between the query system based on relational operations
with RDTs and the system based on evaluating formulas of G\"odel logic, we conclude
that every query formulated by a formula of G\"odel logic is invariant to ordinal
transformations. This observation is a direct consequence of Theorem~\ref{th:inv1}
and Theorem~\ref{th:inv2} and is summarized in the following corollary.

\begin{corollary}\label{col:inv_calculus}
  Let $\varphi$ be a formula of language $\mathcal{J}$ and $\mathbf{M}$
  be a finite $\mathbf{L}$-structure for $\mathcal{J}$.
  Furthermore, let $f$ be an order embedding such that $f(0) = 0$
  and $f(1) = 1$. Then,
  \begin{align}
    \mathcal{D}_{\mathbf{M},\varphi} \circ f &= 
    \mathcal{D}_{\mathbf{M} \circ f,\varphi},
  \end{align}
  where $\mathbf{M} \circ f$ is a finite $\mathbf{L}$-structure
  for $\mathcal{J}$ such that $\rsym^{\mathbf{M} \circ f} = \rsym^{\mathbf{M}} \circ f$
  for any relation symbol $\rsym$ of the language $\mathcal{J}$.
  \qed
\end{corollary}

We now turn our attention to the axiomatization of G\"odel logic and its
consequences for the query systems.
G\"odel logic has a complete Henkin-style axiomatization, i.e.,
a special deductive system. The axiomatization
can be used to find proofs of properties of relational operations owing to
the relationship between the two query systems considered in this section.
The deductive system (for the language $\mathcal{J}$) consists of the
following axioms of logical connectives:
\begin{align}
  & \varphi \logimp (\varphi \logand \varphi), \label{eqn:A1} \\
  & (\varphi \logimp \psi) \logimp
  ((\psi \logimp \chi) \logimp (\varphi \logimp \chi)), \\
  & (\varphi \logand \psi) \logimp \varphi, \\
  & (\varphi \logand \psi) \logimp (\psi \logand \varphi), \\
  & (\varphi \logimp (\psi \logimp \chi)) \logimp
  ((\varphi \logand \psi) \logimp \chi), \\
  & ((\varphi \logand \psi) \logimp \chi) \logimp
  (\varphi \logimp (\psi \logimp \chi)), \\
  & ((\varphi \logimp \psi) \logimp \chi) \logimp
  (((\psi \logimp \varphi) \logimp \chi) \logimp \chi), \\
  & \False \logimp \varphi,
\end{align}
where $\varphi,\psi,\chi$ are arbitrary formulas of $\mathcal{J}$.
In addition to the logical axioms,
we admit the following axioms of substitution
\begin{align}
  &(\forall x)\varphi \logimp \varphi(x/y), \\
  &\varphi(x/y) \logimp (\exists x)\varphi,
\end{align}
where $x$ and $y$ are variables such that $y$ is free for $x$ in $\varphi$
in the usual sense, i.e., no free occurrence of $x$ in $\varphi$ lies within
the scope of a quantifier which binds $y$, see~\cite{Me87}. Furthermore,
we assume the following axioms of the distribution:
\begin{align}
  &(\forall x)(\varphi \logimp \psi) \logimp (\varphi \logimp (\forall x)\psi),
  \\
  &(\forall x)(\psi \logimp \varphi) \logimp ((\exists x)\psi \logimp \varphi),
  \\
  &(\forall x)(\varphi \logor \psi) \logimp (\varphi \logor (\forall x)\psi),
  \label{eqn:AQ3}
\end{align}
where $\varphi$ is an arbitrary formula such that $x$ is not free in $\varphi$.
In addition to the axioms~\eqref{eqn:A1}--\eqref{eqn:AQ3},
the deductive system consists of
deduction rules \emph{modus ponens} ``from $\varphi$ and $\varphi \logimp \psi$
infer $\psi$'' (i.e., the law of detachment) and \emph{generalization}
``from $\varphi$ infer $(\forall x)\varphi$''. As usual, a proof by
a set $\Sigma$ of formulas is a finite sequence $\varphi_1,\ldots,\varphi_n$
where each $\varphi_i$ is a logical axiom or a formula in $\Sigma$ or it
is derived by modus ponens or generalization from preceding formulas in
the sequence; $\varphi$ is provable by $\Sigma$, denoted $\Sigma \vdash \varphi$,
if there is a proof $\varphi_1,\ldots,\varphi_n$ by $\Sigma$ such
that $\varphi = \varphi_n$.

The notion of provability is one paricular notion of (a syntactic) entailment
in the logic. Other notion of entailment---the semantic entailment may be defined
based on the notion of an $\mathbf{L}$-model. In particular, for $\varphi$ and
a safe $\mathbf{L}$-structure $\mathbf{M}$, we put
\begin{align}
  ||\varphi||_\mathbf{M} &=
  \inf\{||\varphi||_{\mathbf{M},v};\, v \text{ is $\mathbf{M}$-valuation}\}.
  \label{eqn:varphi_sem}
\end{align}
Furthermore, a safe $\mathbf{L}$-structure $\mathbf{M}$ is called
a model of $\Sigma$ whenever $||\varphi||_\mathbf{M} = 1$ for all
$\varphi \in \Sigma$. We put $\Sigma \models \varphi$ and say that $\varphi$
is semantically entailed by $\Sigma$ whenever $||\varphi||_{\mathbf{M}} = 1$
for any $\mathbf{L}$-model $\mathbf{M}$ of $\Sigma$ where $\mathbf{L}$
is any totally ordered G\"odel algebra. For convenience, we write
$\vdash \varphi$ and $\models \varphi$ in case of $\Sigma = \emptyset$.
The following completeness theorem is established (recall that
$[0,1]_\mathbf{G}$ denotes the standard G\"odel algebra defined on the real
unit interval).

\begin{theorem}[Completeness of first-order G\"odel logic]\label{th:compl_God}%
  Let $\Sigma$ be any set of formulas of $\mathcal{J}$.
  The following are equivalent:
  \begin{enumlist}\parskip=\smallskipamount%
  \item[\itm{1}]
    $\Sigma \vdash \varphi$;
  \item[\itm{2}]
    $\Sigma \models \varphi$;
  \item[\itm{3}]
    $||\varphi||_\mathbf{M} = 1$ for each $[0,1]_\mathbf{G}$-model of\/ $\Sigma$;
  \item[\itm{4}]
    For each $[0,1]_\mathbf{G}$-structure $\mathbf{M}$ there is $\psi \in \Sigma$
    such that $||\psi||_\mathbf{M} \leq ||\varphi||_\mathbf{M}$;
  \item[\itm{5}]
    For each $[0,1]_\mathbf{G}$-structure $\mathbf{M}$ and each $a \in [0,1]$:
    \\[2pt]
    if $||\psi||_\mathbf{M} \geq a$ for each $\psi \in \Sigma$,
    then $||\varphi||_\mathbf{M} \geq a$.
  \end{enumlist}
\end{theorem}
\begin{proof}
  See~\cite[Theorem 5.2.9 and Corollary 5.3.4]{Haj:MFL}.
\end{proof}

\begin{remark}
  Note that the term ``completeness'' in Theorem~\ref{th:compl_God} refers to
  the syntactico-semantical completeness of the logic, cf. also~\cite{CiHa:Tnbpfl},
  and \emph{not} the functional completeness.
  In fact, the system of connectives used in
  the logic cannot be functionally complete because $[0,1]_\mathbf{G}$
  admits uncountably many $n$-ary functions while the language of the logic and,
  therefore, the number of different formulas that can be written in the
  language, is countable. In this sense, the underlying logic of the rank-aware
  model depart from the classic logic where any $n$-ary function on $\{0,1\}$
  is expressible using (the truth functions of) the fundamental
  connectives (e.g., $\logimp$ and $\logneg$). Also note that G\"odel
  logic is indeed weaker than the classic logic. For instance,
  $\logneg \logneg \varphi \logimp \varphi$ is not provable in
  G\"odel logic. As a consequence, the relational operations with RDTs
  considered in our paper do not satisfy all laws that are satisfied in
  the classic relational model. For instance, there are $\mathcal{D}_1$
  and $\mathcal{D}_2$ on the same relation scheme such that 
  $\mathcal{D}_1 \cap \mathcal{D}_2 \ne
  \minus{\mathcal{D}_1}{(\minus{\mathcal{D}_1}{\mathcal{D}_2})}$.
\end{remark}

As an application of the established connection between the relational operations
with RDTs and G\"odel logic, we can introduce a derived operation of a semijoin.
In the classic mode, a semijoin of $\mathcal{D}_1$ and $\mathcal{D}_2$ on $R_1$
and $R_2$, respectively, may be introduced by
$\project{R_1}{\join{\mathcal{D}_1}{\mathcal{D}_2}}$ or, equivalently,
by $\join{\mathcal{D}_1}{\project{R_1 \cap R_2}{\mathcal{D}_2}}$. From the
perspective of G\"odel logic, $\project{R_1}{\join{\mathcal{D}_1}{\mathcal{D}_2}}$
can be represented by an $\mathbf{L}$-stricture $\mathbf{M}$ with $\rsym_1^\mathbf{M}$
and $\rsym_2^\mathbf{M}$ corresponding to $\mathcal{D}_1$ and $\mathcal{D}_2$,
respectively, and a formula
\begin{align}
  (\exists z_1)\cdots(\exists z_k)
  (\rsym_1(x_1,\ldots,x_i,y_1,\ldots,y_j) \logand
  \rsym_2(y_1,\ldots,y_j,z_1,\ldots,z_k)).
  \label{eqn:derived_semijoin_1}
\end{align}
In G\"odel logic, the formula is equivalent to
\begin{align}
  \rsym_1(x_1,\ldots,x_i,y_1,\ldots,y_j) \logand
  (\exists z_1)\cdots(\exists z_k)\rsym_2(y_1,\ldots,y_j,z_1,\ldots,z_k).
  \label{eqn:derived_semijoin_2}
\end{align}
This is a direct consequence of the fact that
\begin{align}
  \vdash
  (\exists x)(\varphi \logand \psi) \logequ (\varphi \logand (\exists x)\psi)
  \label{eqn:ex_and}
\end{align}
provided that $x$ is not free in $\varphi$, \cite[Lemma 5.1.21]{Haj:MFL}.
Observe that~\eqref{eqn:derived_semijoin_2} is in a correspondence with
$\join{\mathcal{D}_1}{\project{R_1 \cap R_2}{\mathcal{D}_2}}$.
Therefore, in our setting, we also have
\begin{align}
  \project{R_1}{\join{\mathcal{D}_1}{\mathcal{D}_2}} =
  \join{\mathcal{D}_1}{\project{R_1 \cap R_2}{\mathcal{D}_2}}
  \label{eqn:semijoin_law}
\end{align}
as in the classic model which allows us to define a semijoin of ranked
data tables by the expression on either side of the equality
in~\eqref{eqn:semijoin_law}. In addition, owing to the observations
in Theorem~\ref{th:inv1}, the semijoin is also invariant
to the ordinal scaling which follows directly by~\eqref{th:join}
and~\eqref{th:project}.

\section{Computational Issues and Relationship to Other Approaches}\label{sec:comp}
The primary interest of our paper is the invariance to ordinary scaling.
In this section, we make a digression and comment on algorithms for evaluating
particular monotone queries and the relationship to other rank-aware approaches.
We show that our observations on the connection of the relational operations
with RDTs and G\"odel logic can be used to derive laws for query transformations.
In addition, we show that the algorithm for computing top-$k$ query
results~\cite{Fa98:CFIfMS} fits well into our formal model. Finally,
we show that the observations on the invariance to ordinal transormations can
also be applied in the approach by~\cite{LiChCh:RQAaOfRTQ}.

One of the crucial aspects of any model of data from the point of view of its
applicability is the possibility to transform general queries to efficient
logical and then physical query plans. In this section, we show that for
a fragment of the discussed relational operations, one can consider similar
transformations of logical query plans, i.e., transformations of expressions
composed of relational operations with RDTs to equivalent expressions which
are more suitable for an efficient execution, as in the ordinary relational model
of data. We focus only on issues which are specific to our model.

First, we consider laws concerning natural join, projections, selections,
and unions and show that our operations with RTDs admit important
transformation laws which are used in the ordinary relational model.

\begin{theorem}\label{th:laws}
  Let $\mathcal{D}_1$ and $\mathcal{D}_2$ be RDTs
  on relation schemes $R_1$ and $R_2$, respectively.
  Then, the following properties hold true.
  \begin{enumlist}\parskip=\smallskipamount%
  \item[\itm{1}]
    If $\rcond\!: \mathrm{Tupl}(R_1 \cup R_2) \to L$ and
    $\rcond_1\!: \mathrm{Tupl}(R_1) \to L$ are restriction conditions
    such that $\rcond_1(r_1) = \rcond(r_1r_2)$ for any $r_1 \in \mathrm{Tupl}(R_1)$
    and $r_2 \in \mathrm{Tupl}(R_2)$ which are joinable, then
    $\restrict{\rcond}{\join{\mathcal{D}_1}{\mathcal{D}_2}} =
    \join{\restrict{\rcond_1\!}{\mathcal{D}_1}}{\mathcal{D}_2}$.
  \item[\itm{2}]
    If $\rcond_1\!: \mathrm{Tupl}(R_1) \to L$ and
    $\rcond\!: \mathrm{Tupl}(R) \to L$ are restriction conditions
    such that $\rcond(r) = \rcond(rr')$ for all $r \in \mathrm{Tupl}(R)$
    and $r' \in \mathrm{Tupl}(R_1 \setminus R)$, then
    $\project{R}{\restrict{\rcond_1\!}{\mathcal{D}_1}} =
    \restrict{\rcond}{\project{R}{\mathcal{D}_1}}$.
  \item[\itm{3}]
    If $R_1 = R_2$ and $R \subseteq R_1$, then
    $\project{R}{\mathcal{D}_1 \cup \mathcal{D}_2} =
    \project{R}{\mathcal{D}_1} \cup \project{R}{\mathcal{D}_2}$.
  \item[\itm{4}]
    If $S \subseteq R \subseteq R_1$,
    then $\project{S}{\project{R}{\mathcal{D}_1}} =
    \project{S}{\mathcal{D}_1}$.
  \end{enumlist}
\end{theorem}
\begin{proof}
  The assertion can be proved by observing formulas of G\"odel logic corresponding
  to the equalities appearing in \itm{1}--\itm{4} and considering the properties
  of $\vdash$ in G\"odel logic. In case of \itm{1},
  $\restrict{\rcond}{\join{\mathcal{D}_1}{\mathcal{D}_2}}$ can be represented by
  a formula $\theta \logand (\varphi \logand \psi)$ and, analogously,
  $\join{\restrict{\rcond_1\!}{\mathcal{D}_1}}{\mathcal{D}_2}$ can be represented
  by a formula $(\theta_1 \logand \varphi) \logand \psi$ (we have tacitly
  identified restriction conditions with formulas). Therefore, \itm{1} follows
  by the associativity of $\logand$, i.e.,
  \begin{align*}
    \vdash (\theta \logand (\varphi \logand \psi)) \logequ
    ((\theta \logand \varphi) \logand \psi),
  \end{align*}
  see \cite[Lemma 2.2.15]{Haj:MFL}, and the relationship of $\rcond_1$ and $\rcond$.
  Analogously, \itm{2} is a consequence of~\eqref{eqn:ex_and};
  \itm{3} is a consequence of
  \begin{align*}
    \vdash (\exists x)(\varphi \logor \psi) \logequ (\varphi \logor (\exists x)\psi)
  \end{align*}
  provided that $x$ is not free in $\varphi$, see~\cite[Lemma 5.1.21]{Haj:MFL}.
  Finally, \itm{4} follows directly by the fact that the left-hand and right-hand sides
  of the equality in \itm{4} translate into a single formula of G\"odel logic of
  the form $(\exists x_1)\cdots(\exists x_n)\varphi$.
\end{proof}

As a consequence of Theorem~\ref{th:laws} and~\eqref{eqn:semijoin_law},
we may conclude that the usual optimization techniques based on 
pushing down restrictions and projections~\cite{GaMoUlWi:DSCB} still work
in the ranke-aware model because the classic laws on which the optimizations
are based are preserved in G\"odel logic. Therefore, many monotone queries in
the rank-aware model can be transformed into expressions of the form
\begin{align}
  \mathcal{D}_1 \bowtie \mathcal{D}_2 \bowtie \cdots \bowtie \mathcal{D}_n,
  \label{eqn:join_nexpr}
\end{align}
where $\mathcal{D}_i$ for $i=1,\ldots,n$ are RDTs on $R_i$ which result
by computing projections and/or restrictions of RDTs representing base
data (i.e., RDTs bound to relation variables in a database instance).
Clearly, a tuple $r_1r_2\cdots r_n$,
where $r_i \in \mathrm{Tupl}(R_i)$ for each $i=1,\ldots,n$,
belongs to the answer set of~\eqref{eqn:join_nexpr} if{}f
all $r_1,r_2,\ldots,r_n$ are joinable and its score is
\begin{align}
  \inf\{\mathcal{D}_1(r_1),\mathcal{D}_2(r_2),\ldots,\mathcal{D}_n(r_n)\} > 0.
\end{align}
Hence, \eqref{eqn:join_nexpr} may be understood as a query in a similar form
as \eqref{eqn:motiv} with a few minor conceptual differences:
(i) $\mathcal{D}_i$ in may not represent a result of an atomic query
as in~\eqref{eqn:motiv}, (ii) we always consider $\inf$ as the interpretation
of $\mathop{\&}$, and (iii) the objects which match queries are in fact tuples
constructed as joins of joinable tuples considered
on general schemes $R_1,R_2,\ldots,R_n$.
Nevertheless, in case one wants to compute only top $k$ matches, i.e.,
if one wants to compute only a portion of the answer
set of~\eqref{eqn:join_nexpr} consisting only of $k$ tuples with highest scores,
we can directly adopt the Fagin algorithm~\cite{Fa98:CFIfMS}, namely,
its improved version which consideres $\inf$ as the aggregation function,
see~\cite[Theorem 4.4]{Fa98:CFIfMS}, provided that each $\mathcal{D}_i$
allows an efficient ``sorted access'' (tuples in the answer set of
$\mathcal{D}_i$ may be listed sorted by scores in the descending order)
which may be assumed in many natural situations. Except for the technical
difference in using ``joinable tuples'', see (iii) above, the Fagin
algorithm does not need to be further modified. Interested readers are
refered to~\cite{Fa98:CFIfMS} for details and complexity analysis.

We now turn our attention to RankSQL and the extended relational algebra
proposed in~\cite{LiChCh:RQAaOfRTQ} which is arguably one of the most
influential approaches to ranking in relational databases.
Similar observation as in Section~\ref{sec:inv} can be made
in the rank-relational approach described in their paper.
Recall that according to~\cite{LiChCh:RQAaOfRTQ}, the basic structure
which serves as a counterpart of the classic relations on relation schemes
is a \emph{rank-relation} $R_\mathcal{P}$ which is understood as a classic
relation $R$ equipped with \emph{scores} and (strict total)
tuple \emph{order} $<_{R_\mathcal{P}}$ based on the scores. The score
for each tuple $r \in R$ is computed as a result of a general (monotonic)
\emph{scoring function} $\mathcal{F}$ which is applied
to \emph{predicate scores} $p_i[r]$ of the tuple $r$.
The predicate scores represent individual ranking criteria
(like low price, high availability, close distance between locations, etc.)
called predicates and denoted by $p_1,\ldots,p_n$. Note that $\mathcal{P}$
(called the \emph{set of evaluated predicates}) is always a subset
of $\{p_1,\ldots,p_n\}$ and the rank-relational model and its implementation
relies on the \emph{ranking principle} \cite[page 133]{LiChCh:RQAaOfRTQ}
based on \emph{maximal possible scores} of tuples in $R$
given $\mathcal{P}$, i.e., each $p_i[r]$ for which $p_i \not\in \mathcal{P}$
($p_i$ is not evaluated) is considered to have the application-specific
maximal possible value of~$p_i$. Therefore, for general $\mathcal{P}$,
each tuple $r \in R$ has its maximal possible score denoted
$\overline{\mathcal{F}}_\mathcal{P}[r]$ and $<_{R_\mathcal{P}}$
(the tuple order in $R$ given $\mathcal{P}$) is introduced based on
such scores, namely, $r_1 <_{R_\mathcal{P}} r_2$ whenever 
$\overline{\mathcal{F}}_\mathcal{P}[r_1] < \overline{\mathcal{F}}_\mathcal{P}[r_2]$.
Furthermore, queries in RankSQL are transformed into expressions
of rank-relational algebra which introduces operations with rank-relations
including restriction, union, intersection, difference, theta-join,
and \emph{rank}---a new operation which produces $R_{\mathcal{P} \cup \{p\}}$
based on $R_{\mathcal{P}}$ and $p \not\in \mathcal{P}$. Let us stress that
the operations with rank-relations indeed operate on rank-relations, i.e.,
based on scores and tuple orders of the input arguments, the operations
define scores and tuple order of the result. For instance, in case of
the union of $R_{\mathcal{P}_1}$ and $S_{\mathcal{P}_2}$, the result is
a rank-relation $(R \cup S)_{\mathcal{P}_1 \cup \mathcal{P}_2}$ in which
$r_1 <_{(R \cup S)_{\mathcal{P}_1 \cup \mathcal{P}_2}} r_2$ whenever
$\overline{\mathcal{F}}_{\mathcal{P}_1 \cup \mathcal{P}_2}[r_1] <
\overline{\mathcal{F}}_{\mathcal{P}_1 \cup \mathcal{P}_2}[r_2]$.

From our perspective, we may view an important special case of the
rank-relational approach in~\cite{LiChCh:RQAaOfRTQ} as follows:
We consider $\mathbf{L}$ (the structure of scores) as a totally
ordered complete lattice and we let $\mathcal{F}$ be $\inf$. That is,
the scoring function always computes the minimum of given predicate scores
and $1$ represents the maximal possible score. In this setting,
rank-relations may be viewed as RDTs with the order of tuples given
implicitly by the scores; predicates $p_i$ may be viewed as general
restriction conditions, and the rank operator may be seen as
a general restriction~\eqref{eqn:restrict}. Moreover, for the
rank-relational querying, we may ask the same question as before:
Does an ordinal transformation of the input ranking criteria yield
the same results? The answer is positive. In a more detail,
let $f\!: L \to L$ be an order embedding which preserves $0$ and $1$.
Then, for each $p_i$ (which is in fact a map from the set of all tuples
on the scheme of $R$ to $L$) we may consider the composed map $p_i \circ f$
and put $\mathcal{P} \circ f = \{p \circ f;\, p \in \mathcal{P}\}$.
With respect to the above-mentioned interpretation of evaluated predicates,
$\mathcal{P} \circ f$ may be seen as a set of ordinally transformed
evaluated predicates. Moreover,$f$ is an order embedding and for $\cup$
defined as above, we have
\begin{align*}
  &
  r_1 <_{(R \cup S)_{\mathcal{P}_1 \cup \mathcal{P}_2}} r_2
  \text{ if{}f}
  \\
  &
  \overline{\mathcal{F}}_{\mathcal{P}_1 \cup \mathcal{P}_2}[r_1] <
  \overline{\mathcal{F}}_{\mathcal{P}_1 \cup \mathcal{P}_2}[r_2]
  \text{ if{}f}
  \\
  &
  f\bigl(\overline{\mathcal{F}}_{\mathcal{P}_1 \cup \mathcal{P}_2}[r_1]\bigr) <
  f\bigl(\overline{\mathcal{F}}_{\mathcal{P}_1 \cup \mathcal{P}_2}[r_2]\bigr).
\end{align*}
Now, using the fact that $\mathcal{F}$ is $\inf$, it follows that
$f\bigl(\overline{\mathcal{F}}_{\mathcal{P}_1 \cup \mathcal{P}_2}[r]\bigr) =
\overline{\mathcal{F}}_{(\mathcal{P}_1 \circ f) \cup (\mathcal{P}_2 \circ f)}[r]$
for all $r \in R$. Hence, by the definition of
$<_{(R \cup S)_{\mathcal{P}_1 \cup \mathcal{P}_2}}$ and 
$<_{(R \cup S)_{(\mathcal{P}_1 \circ f) \cup (\mathcal{P}_2 \circ f)}}$,
\begin{align*}
  &
  r_1 <_{(R \cup S)_{\mathcal{P}_1 \cup \mathcal{P}_2}} r_2
  \text{ if{}f }
  \\
  &
  \overline{\mathcal{F}}_{(\mathcal{P}_1 \circ f) \cup (\mathcal{P}_2 \circ f)}[r_1] <
  \overline{\mathcal{F}}_{(\mathcal{P}_1 \circ f) \cup (\mathcal{P}_2 \circ f)}[r_2]
  \text{ if{}f}
  \\
  &
  r_1 <_{(R \cup S)_{(\mathcal{P}_1 \circ f) \cup (\mathcal{P}_2 \circ f)}} r_2,
\end{align*}
proving that $(R \cup S)_{\mathcal{P}_1 \cup \mathcal{P}_2}$ and
$(R \cup S)_{(\mathcal{P}_1 \circ f) \cup (\mathcal{P}_2 \circ f)}$
are ordinally equivalent. One may proceed the same way for the other
operations of the rank-relation algebra, see~\cite[page 134]{LiChCh:RQAaOfRTQ}.
As a consequence, ordinal transformations do not have any effect on
the results of top-$k$ queries---scores of tuples may be different,
however, the order in which tuples appear in the result is the same.

Finally, let us note that the approach in~\cite{LiChCh:RQAaOfRTQ} is conceptually
similar to ours in that both are capable to answer queries by relations
with tuples annotated by scores which indicate degrees of matches of user
preferences. It should be noted, however, that the approaches are technically
different (even if we consider $\mathcal{F}$ as $\inf$). More detailed on the
technical differences can be found in~\cite{VaVy:Rdrad}.

\section{Conclusion}
Notions of ordinal containment and ordinal equivalence of relations consisting of
tuples annotated by scores have been proposed. The ordinal containment and
equivalence have been characterized in terms of the existence of suitable
order-preserving functions and order isomorphisms between subsets of scores.
It has been shown that infima-based algebraic operations with ranked relations
are invariant to ordinal transformations: Queries applied to original
and transformed data yield results which are equivalent in terms of the order
given by scores. We have demonstrated that this property is not preserved if
one considers algebraic operations with ranked relations based on general
aggregation functions like triangular norms (other than the minimum triangular
norm). As a result of our observation, we have concluded that under infima-based
algebraic operations, the scores in ranked tables have no quantitative meaning.
Generality of the result has been demonstrated by applying the observations
in an alternative calculus-based query system grounded in G\"odel logic.
Furthermore, relationship to other approaches has been investigated with
the intention to show the connection to existing algorithms for monotone
query evaluation and conceptually similar approaches to ranking in databases.

\subsubsection*{Acknowledgment}
Supported by grant no. \verb|P202/14-11585S| of the Czech Science Foundation.

\footnotesize
\bibliographystyle{amsplain}
\bibliography{iotrad}

\end{document}